\documentclass[pra,amsmath,amssymb,amsfonts,nofootinbib]{revtex4}

\usepackage{bm,graphicx,mathrsfs,color}

\usepackage{amsmath,bbm}
\usepackage{amsfonts,amssymb}
\usepackage{graphicx}
\usepackage{epsfig}
\usepackage{times}
\usepackage{verbatim}%\textbf{}
\usepackage{enumerate}

\newcommand{\be}{\begin{equation}}
\newcommand{\ee}{\end{equation}}
\newcommand{\bea}{\begin{eqnarray}}
\newcommand{\eea}{\end{eqnarray}}
\newcommand{\ba}{\begin{align}}
\newcommand{\ea}{\end{align}}

\newcommand{\1}{\mathbbm{1}}
\newcommand{\id}{{\rm id}}

\newcommand{\ket}[1]{\left | \, #1 \right\rangle}
\newcommand{\bra}[1]{\left \langle #1 \, \right |}

\newcommand{\avr}[1]{\left \langle#1 \right \rangle}

\newcommand{\tr}[1]{{\rm tr}\left[{#1}\right]}
\newcommand{\Tr}{\rm tr}

\newcommand{\raw}{\rightarrow}

\newcommand{\bR}{\mathbbm{R}}
\newcommand{\bN}{\mathbbm{N}}

\newcommand{\bL}{\mathbbm{L}}
\newcommand{\bE}{\mathbbm{E}}

\newcommand{\bZ}{\mathbbm{Z}}

\newcommand{\cH}{\mathcal{H}}
\newcommand{\cB}{\mathcal{B}}
\newcommand{\cR}{\mathcal{R}}

\newcommand{\cD}{\mathcal{D}}
\newcommand{\cQ}{\mathcal{Q}}
\newcommand{\cL}{\mathcal{L}}
\newcommand{\cK}{\mathcal{K}}
\newcommand{\cA}{\mathcal{A}}
\newcommand{\cS}{\mathcal{S}}
\newcommand{\cW}{\mathcal W}

\newcommand{\half}{\frac{1}{2}}

\newcommand{\Ent}{{\rm Ent}}
\newcommand{\Var}{{\rm Var}}
\newcommand{\Cov}{{\rm Cov}}

\newtheorem{theorem}{Theorem}
\newtheorem{lemma}[theorem]{Lemma}

\newtheorem{corollary}[theorem]{Corollary}
\newtheorem{proposition}[theorem]{Proposition}
\newtheorem{definition}[theorem]{Definition}

\def\Proof{\noindent\textsc{Proof:}}
\def\proof{\Proof}
\def\qed{\leavevmode\unskip\penalty9999 \hbox{}\nobreak\hfill
     \quad\hbox{\leavevmode  \hbox to.77778em{%
               \hfil\vrule   \vbox to.675em%
               {\hrule width.6em\vfil\hrule}\vrule\hfil}}
     \par\vskip3pt}
    {\hspace*{\fill}$\Box$\vspace{1.5ex}\par}

\newcommand{\mjk}[1]{\textcolor{black}{#1}}

\begin{document}

\title{\sc{\Large Quantum Gibbs samplers: the commuting case}}

\author{Michael J. Kastoryano$^{1}$  and Fernando G. S. L. Brand\~ao$^{2, 3}$}
\address{
$^1$ Dahlem Center for Complex Quantum Systems, Freie Universit\"at Berlin, 14195 Berlin, Germany\\
$^2$ Quantum Architectures and Computation Group, Microsoft Research, Redmond, US  \\
$^3$ Department of Computer Science, University College London, London UK 
}

\date{\today}

\begin{abstract} 
We analyze the problem of preparing quantum Gibbs states of lattice spin Hamiltonians with local and commuting terms on a quantum computer and in nature. Our central result is an equivalence between the behavior of correlations in the Gibbs state and the mixing time of the semigroup which drives the system to thermal equilibrium (the Gibbs sampler). We introduce a framework for analyzing the correlation and mixing properties of quantum Gibbs states and quantum Gibbs samplers, which is rooted in the theory of non-commutative $\bL_p$ spaces. We consider two distinct classes of Gibbs samplers, one of them being the well-studied Davies generator modelling the dynamics of a system due to weak-coupling with a large Markovian environment. We show that their spectral gap is independent of system size if, and only if, a certain strong form of clustering of correlations holds in the Gibbs state. Therefore every Gibbs state of a commuting Hamiltonian that satisfies clustering of correlations in this strong sense can be prepared efficiently on a quantum computer. As concrete applications of our formalism, we show that for every one-dimensional lattice system, or for systems in lattices of any dimension at temperatures above a certain threshold, the Gibbs samplers of commuting Hamiltonians are always gapped, giving an efficient way of preparing the associated Gibbs states on a quantum computer.

\end{abstract}
\maketitle

\tableofcontents
%%%%%%%%%%%%%%%%%%%%%%%%%%%%%%%%%%%%%%%%%%%%%%%%%%%%%%%%%%%%%%%%%%%%%
\section{Introduction}

Physical systems very often are in thermal equilibrium. Statistical mechanics provides a microscopic theory justifying the relevance of thermal states of matter. However, fully understanding the ubiquity of this class of states from the laws of quantum theory remains an important topic in theoretical physics. Indeed the problem of thermalization in quantum systems has recently generated a lot of interest, in part because of the new set of tools available from the field of quantum information theory \cite{thermalization1,thermalization2,thermalization3, thermalization4}. The problem can be broken up into two sets of questions: (i) under what conditions does a system thermalize in the long time limit, and (ii) assuming a system does eventually thermalize, how much time does one have to wait before thermalization is reached? Our work is concerned with the latter question in the setting of quantum lattice spin systems.  

The problem of the time of thermalization is also of practical relevance in the context of quantum simulators, where one wants to emulate the behavior of a real physical system by simulating an idealization of it on a classical or quantum computer. Given that many of the systems which one would want to simulate are thermal, it is an important task to develop simulation and sampling algorithms that can prepare large classes of thermal states of physically relevant Hamiltonians. A large body of work has already been done on the classical problem, starting with the development and analysis of Gibbs sampling algorithms called Glauber dynamics, which include the Metropolis and Heat-bath algorithms as spacial cases. Nowadays, there are dozens of variants of classical Gibbs samplers which find applications in a variety of fields of physics, computer science, and theoretical chemistry \cite{MonteCarloBook}.  A common feature of many of these algorithms is that they provide reliable results in practice, but a systematic certification of their accuracy and efficiency is often elusive. Although a very hard problem in general, estimating the convergence time of classical Gibbs samplers has seen a number of breakthroughs in the past few decades. These methods, which are rooted in the analysis of Markov Chain mixing times, are also closely related to problems in combinatorial optimization, with applications in numerous fields; see \cite{Diaconis,Peres} for recent surveys.

The rigorous analysis of classical Gibbs sampling algorithms was tied to the development of a rigorous theory of thermodynamic quantities, which consists in analyzing the behavior of finite systems of larger and larger size in order to infer the physics in the thermodynamic limit. This is made possible because the influence of the infinite system on a subregion of the lattice can be encoded in the boundary of the region in a one-to-one fashion. This identity is often called \textit{Dobrushin-Lanford-Ruelle} (DLR) theory \cite{Dobrushin,Ruelle}, and is the technical cornerstone of a lot of the early rigorous results in lattice Gibbs states. Using DLR theory, it has been possible to properly characterize phase transitions in classical lattice spin models. A series of seminal results have shown that the existence of a unique Gibbs state in the thermodynamic limit is related both to exponentially decaying correlations between local observables and rapidly mixing dynamics to thermal equilibrium \cite{MartinelliBook,MartinelliReview,WeakMixingStrongMixing,Guionnet,Yoshida1,Yoshida2,Cesi,Bertini}.

The purpose of this work is to introduce a framework for analyzing Gibbs samplers in the setting of quantum systems, and to explore to what extent the classical results (equivalences) generalize to quantum lattice spin systems. We build upon a growing body of work on the theory of mixing times of quantum channels \cite{chi2,LogSobolev}. 

Throughout this work we will restrict ourselves to commuting Hamiltonians. It is worth noting that the case of commuting Hamiltonians does not effectively reduce to the classical problem, as highly non-classical phenomenon such as topological quantum order can occur. In particular, this setting encompasses all stabilizer Hamiltonians, which have been a useful playground for exploring the unique quantum features of many-body systems.

The relevance of the problem we analyse is two-fold. First, we consider a class of Gibbs samplers, called Davies generators \cite{Dav79}, which can be derived from the weak coupling of a finite quantum system to a large thermal bath. Hence our analysis pertains to the time it takes to reach thermal equilibrium in naturally occurring systems. Secondly, all Gibbs samplers which we consider are local and bounded maps, and therefore can be prepared by dissipative engineering or digital simulation on quantum computers, or quantum simulators \cite{KBGKE11}. Thus, our results can also be understood as an analysis of the efficiency of quantum Gibbs sampling.

Previous efforts at proposing quantum Gibbs samplers have typically fallen into two categories: those that have a certified runtime, but that is exponential in the system size \cite{GibbsSampler1,GibbsSampler2,GibbsSampler3}, and those that can be implemented on a quantum computer efficiently, but where no bounds on their convergence are currently known \cite{Q_Metropolis, YAG12, ORR12}. Our Gibbs samplers have the benefit that they are very simple, and hence amenable to analytic study. We reiterate however that they have the drawback of only being properly defined for Hamiltonians with local commuting terms.

\subsection{Summary of results}

In order to present our main results, we need to spend some time defining the framework and the quantities involved, which are rooted in the theory of non-commutative $\bL_p$ spaces. This has value in its own right, as a systematic study of thermal states of quantum lattice systems in the spirit of DLR theory has not yet been undertaken. However we only achieve it partially, and comment along the way on the limitations of a full generalization of the classical theory.

After a brief recollection of the formal framework and of the setting of dissipative quantum systems, we introduce a class of maps called \text{conditional expectations} which serve as local quasi-projectors onto the Gibbs state of the system. These maps play a central role in our analysis. We identify two special classes of conditional expectations: the first is purely dynamical and inherits many of the properties of the underlying dissipative generator; the second is purely static, and only depends on the reference (Gibbs) state of the system. We prove that both are local maps when the underlying Hamiltonian is commuting (see Propositions \ref{thm:liouvilleprojector} and \ref{thm:localprojector}).

We go on to define quantum lattice Gibbs states, and introduce two classes of Gibbs samplers: \textit{Davies Generator} and \textit{Heat-Bath Generator}. We will also call them Davies Gibbs sampler and Heat-Bath Gibbs sampler, respectively. The Davies generator is obtained from a canonical weak-coupling between a system and a large thermal bath, whereas the heat-bath generator is constructed in a manner reminiscent of the classical heat-bath Monte-Carlo algorithm. The basic properties of these maps are summarized and collected in Propositions \ref{lem:DaviesGen} and \ref{lem:HBGen}. Both have the important property that the associated generators are local (see Section \ref{sec:gibbsSS}). As a result, they can be implemented efficiently on a quantum computer using the method of Ref. \cite{KBGKE11}.

The main purpose of the paper is to show an equivalence between the convergence time of the Gibbs sampler and the correlation behaviour of the Gibbs state. The analogous classical equivalence builds very heavily on the DLR theory of boundary conditions.  As a naive extension of the DLR theory does not hold for quantum systems \cite{WernerBC}, we are lead to define a different notion of clustering (which we call \textit{strong clustering}), that incorporates the \textit{strong mixing}  (or complete analyticity) condition for classical systems. This condition relies on a conditional covariance, which restricts attention to a subset of the lattice. The conditional covariance depends on a specific choice of conditional expectation. We show that the strong clustering condition implies the standard clustering of correlation (which we call \textit{weak clustering}) condition that is usually considered in quantum lattice systems (see Corollary \ref{cor:weakcluster}). We also flesh out the connection between our notions of clustering and the local indistinguishability of two Gibbs states differing only by a distant perturbation (\ref{StrongCLusLTQO}). 

Having introduced the framework of Gibbs samplers and defined what we mean by clustering of correlations, we set out to prove the main theorem of our paper (see Theorems \ref{thm:main} and \ref{thm:converse} for a precise formulation):

\begin{theorem}[informal]  \label{thminformal1}
Both the Davies Gibbs sampler and Heat-Bath Gibbs sampler of commuting local Hamiltonians have a gap independent of the system size if, and only if, the Gibbs state satisfies strong clustering. 
\end{theorem}

The gap is defined in Def. \ref{def:gap} and is related to the rate of convergence of the Gibbs sampler to the Gibbs state. Since both Gibbs sampler we consider can be implemented efficiently on a quantum computer, we find:
\begin{corollary}
Any Gibbs state satisfying strong clustering can be prepared on a quantum computer in polynomial time (in the number of sites of the Hamiltonian)
\end{corollary}

We prove the necessity and sufficient parts of Theorem \ref{thminformal1} separately, as they require distinct proof techniques. On one hand, the \textit{only if} statement is proved via methods reminiscent of the analogous classical result \cite{Bertini}. The main idea of the proof is to consider the variational characterization of the spectral gap, and show, by a careful manipulation of conditional variances, that the gap of the Gibbs sampler restricted to a subsystem of minimum side length $L$ is roughly the same as the gap restricted to  a subsystem of side length $2L$. Then using the same argument iteratively shows that the gap of the dynamics is approximately scale invariant. The \textit{if} part of the statement, on the other hand, exploits methods from quantum information theory and quantum many-body theory. In particular, we find a mapping of our problem to properties of frustration-free gapped local Hamiltonians, and apply the so-called detectability lemma of \cite{DetectLemma2}. 

Our main theorem  becomes especially compelling for one dimensional lattice systems, where it was shown by Araki \cite{Araki} that Gibbs states always cluster. We prove that weak clustering and strong clustering are equivalent for one dimensional systems, getting that all Gibbs samplers in this case are gapped. Exploring our mapping between Gibbs samplers and local Hamiltonians, we also prove that at high enough temperature (independent of the size of the system) the Gibbs samplers are gapped. We then obtain:

\begin{theorem}[informal] 
Both the Davies and the Heat-Bath Gibbs samplers give polynomial-time quantum algorithms for preparing the Gibbs state of every 1D commuting Hamiltonian at any constant temperature and every commuting Hamiltonian at temperatures above a critical temperature (that is independent of the system size).
\end{theorem}

We note that since Gibbs states of 1D commuting Hamiltonians are matrix-product operators, one can prepare them efficiently on a quantum computer using e.g. \cite{SSVCW05} (in fact this is also true for general non-commuting 1D Gibbs states \cite{Has05a}); here we only show another way of preparing them, which might be more resilient to noise in some circumstances (as the preparation can be achieved for example simply by connecting the system to a large heat bath).

Finally, we discuss extensions and further implications of our results. We conclude with a few important questions and conjectures, and connect some of them with the problem of self-correcting quantum memories in two dimensional systems.

\section{Preliminaries}

\subsection{Notation}

This paper concerns quantum spin lattice systems. Although the results presented here can be extended to more general graphs, we will restrict our attention to spins living on a $d$-dimensional finite square lattice $\Lambda \subseteq \bZ^d$, which can be identified with $(\bZ \backslash l)^d$ for an integer $l$ (we will call $l$ the lattice side length). Lattice subsets will be denoted by upper case Latin letters, e.g. $A,B\subseteq\Lambda$. The complement of a set $A\in\Lambda$ will be written $A^c$. The cardinality of a set $A$ will be denoted by $|A|$. 

The global Hilbert space associated to $\Lambda$ is $\cH_\Lambda=\bigotimes_{x\in \Lambda}\cH_x$. We assume the local Hilbert spaces are finite dimensional, i.e. ${\rm dim}(\cH_x)<\infty$. We denote the set of bounded operators acting on $\cH_\Lambda$ by $\cB_\Lambda\equiv \cB$, and its Hermitian subset by $\cA_\Lambda\subseteq\cB_\Lambda$. The set of positive semi-definite operators is denoted $\cS_\Lambda=\{X \in \cA_\Lambda,X\geq0,\tr{X}=1\}$, and its full rank subset is denoted by $\cS_\Lambda^+$. The elements of $\cA_\Lambda$ will be called observables, and will always be denoted with lower case Latin letters ($f,g\in\cA_\Lambda$). The elements of $\cS_\Lambda$ will be called states, and will be denoted with lower case Greek letters ($\rho,\sigma\in\cS_\Lambda$). 

We will say, in a a slight abuse of notation, that an operator $f \in \cB_\Lambda$ has support on $A \subset \Lambda$ if it can be written as $f_{A} \otimes \id_{A^c}$, for an operator $f_A \in \cB_A$. We will denote the support of a local observable $f\in\cA_\Lambda$ by $\Sigma_f$. 

We will make use of the modified partial trace $\Tr_A: \cB_\Lambda \rightarrow \cB_\Lambda$, which would read the following in the more traditional quantum information theory notation: $\Tr_A:f\mapsto\Tr_A(f)\otimes \1_A$. 

For $i,j\in\Lambda$, we denote by $d(i,j)$ the Euclidean distance in $\bZ^d$. The distance between two sets $A,B\subseteq\Lambda$ is $d(A,B) :=\min\{d(i,j),i\in A,j\in B\}$.

\subsection{$\bL_p$ spaces}

Many of the tools used in this work are rooted in the theory of non-commutative $\bL_p$ spaces \cite{LpSpaces1,LpSpaces2}. The central property of the non-commutative $\bL_p$ spaces summarized below, is that the norm as well as the scalar product is weighted with respect to some 
full rank reference state $\rho \in \cS_\Lambda^+$. The non-commutative $\bL_p$ spaces are equipped with a weighted $\bL_p$-norm which, for any $f,g\in\cA_\Lambda$ and some fixed $\rho\in\cS_\Lambda^+$, is defined as
\be 
\| f \|_{p,\rho} :=\tr{\;|\; \rho^{\frac{1}{2p}} f \rho^{\frac{1}{2p}}\;|^p\;}^{\frac{1}{p}}, ~~~~~~~~~ 1\leq p\leq \infty.
\ee
Similarly, the $\rho$-weighted non-commutative $\bL_p$ inner product is given by 
\be
\avr{f,g}_\rho := \tr{\sqrt{\rho} f \sqrt{\rho} g}.
\ee

We summarize the basic properties of non-commutative $\bL_p$ spaces in the following lemma:

\begin{lemma}{\cite{LpSpaces1,LpSpaces2}}\label{Lem:Lp-norm}
Let $\rho\in\cS_\Lambda^+$. Then
\begin{enumerate}
\item  (Natural ordering) Let $f\in \cA_\Lambda$. Then for any $p,q \in [1,\infty)$ satisfying $p\leq q$, we get $||f||_{p,\rho}\leq||f||_{q,\rho}$. 

\item (H\" older-type inequality) Let $f,g\in \cA_\Lambda$. Then for any $p,q\in[0,\infty)$ satisfying $1/p+1/q=1$, 
\be |\avr{f,g}_\rho|\leq ||f||_{p,\rho}||g||_{q,\rho}\ee

\item (Duality) Let $f\in \cA_\Lambda$. Then for any $p,q\in[0,\infty)$ satisfying $1/p+1/q=1$, 
\be ||f||_{p,\rho}=\sup\{\avr{g,f}_\rho, g\in\cA_\Lambda, ||g||_{q,\rho}\leq 1\}.\ee

\end{enumerate}
\end{lemma}

In the remainder of the paper, unless specified otherwise, we will always be working with $\bL_p$ norms and inner products. The reference state should always be clear from the context, and will almost always be the Gibbs state of a local commuting Hamiltonian (see Sec. \ref{sec:gibbsSS}).

Finally, we will also make extensive use of the $\bL_p$ covariance of a state $\rho\in\cS_\Lambda$, which is defined for any $f,g\in\cA_\Lambda$ as 
\be
\Cov_\rho(f,g) = \left|\avr{f,g}_\rho - \tr{\rho f }\tr{\rho g }\right|. 
\ee 
Similarly, the variance is given as $\Var_\rho(f)=\Cov_\rho(f,f)$. 
The covariance and the variance are always positive, and  they are invariant under the transformation $g\raw g + c\1$, for any $c\in\bR$.

\subsection{Dynamics}

The time evolution of an observable ($f_t\in\cA_\Lambda$) will be described by one-parameter semigroups of completely positive unital maps $T_t := e^{t \cL}$, whose generator $\cL : \cB_\Lambda \rightarrow \cB_\Lambda$ (also called the Liouvillian) can always be written in standard \textit{Lindblad form} (see Chapter 7 of \cite{Wolf12} for a derivation): 
\be
\partial_t f =  \cL(f_t) \equiv i[H,f] + \cD(f),
\ee
with
\be
\cD(f)=\sum_i L^\dag_i f L_i - \frac{1}{2}\{L^\dagger_i L_i , f\}_+,
\ee
where $\{L_i\} \in  \cB_\Lambda$ are Lindblad operators, $H \in \cA_\Lambda$ is a Hamiltonian operator and $\{A, B \}_+ := AB + BA$ is the anti-commutator. This evolution corresponds to the Heisenberg picture, which specifies the dynamics on observables rather than on states. We denote the dual generator, with respect to the Hilbert-Schmidt inner product $\langle X, Y \rangle := \text{tr}(X Y^{\cal y})$, by $\cL^*$; it amounts to the evolution of states, i.e. the Schr\"odinger picture. The fact that $T_t$ is unital for every $t$ ensures that $\cL(\1)=0$. 

A Liouvillian $\cL:\cB_\Lambda\rightarrow\cB_\Lambda$  is said to be \textit{primitive} if it has a unique full-rank stationary state (i.e. a unique full rank sate $\rho$ s.t. $\cL^*(\rho)=0$). A Liouvillian is said to be \textit{reversible} (or satisfy \textit{detailed balance}) with respect to a state $\rho\in\cS_\Lambda$ if for every $f,g\in\cA_\Lambda$, 
\be \avr{f,\cL(g)}_\rho=\avr{\cL(f),g}_\rho.\ee
If $\cL$ is reversible with respect to $\rho$ then $\rho$ is a stationary state of $\cL$ \cite{LogSobolev}. Note that for a classical Hamiltonian, the condition above reduces to the usual detailed balance condition for the generator of a Markov chain.

For an integer $r$ and lattice $\Lambda$, $\cL = \cL_{\Lambda}$ is said to be \textit{r-local} if it can be written as
\be \cL(f)=\sum_{Z \subset \Lambda} i[H_Z,f] +\cD_Z(f),\label{eqn:liouvgen}\ee
where $H_Z$ and $\cD_Z$ only have support on subsets $Z\subseteq\Lambda$ such that $|Z| \leq  r$. The integer $r$ will be referred to as the \textit{range} of the Liouvillian.  

Let us recall the main result of \cite{KBGKE11}, which shows that for a $r$-local Liovillian $\cL$ acting on $\cB( {\cal H}_{\Lambda} )$, with ${\cal H}_{\Lambda} = \bigotimes_{k \in \Lambda} {\cal H}_{k}$, $d := \max_{k} \text{dim}({\cal H}_k)$, and $K$ local terms ${\cal D}_Z$, $e^{t \cL}$ can be implemented with error $\varepsilon$ on a quantum computer by a circuit of size $O(K^3 t^2 \varepsilon^{-1})$ (with the error measured in terms of the diamond norm).

We will also consider restricted Liouvillians of $\cL$ defined only by the terms acting on a subset $A\subseteq\Lambda$:
\be  \cL_A(f)=\sum_{Z\cap A\neq0} i[H_Z,f] +\cD_Z(f)\nonumber\ee
Note that $\cL_A$ acts on $A$ plus a neighbourhood of $A$ whose radius is given by the range of the local terms (the parameter $r$) and the geometry of the lattice.  

We now define a few concepts for local Liouvillians that will show useful later.

\begin{definition}
Let $\Lambda$ be a finite subset of $\bZ^d$. Let $\cL:\cB_\Lambda\rightarrow\cB_\Lambda$ be a local Liouvillian. We say that $\cL$ is locally primitive if for any subset $A\subseteq\Lambda$, $\cL_A(f)=0$ implies that $f$ acts non-trivially only on $A^c$. Similarly, $\cL_\Lambda$ is locally reversible with respect to $\rho$ if for any $A\subseteq\Lambda$ and every $f,g\in\cA_\Lambda$, 

\be \avr{f,\cL_A(g)}_\rho=\avr{\cL_A(f),g}_\rho\nonumber\ee
\end{definition}

\begin{definition}
Let $\Lambda$ be a finite subset of $\bZ^d$. Let $\cL :\cB_\Lambda\rightarrow\cB_\Lambda$ be a local Liouvillian. We say that $\cL$ is frustration free if for any $A\subseteq\Lambda$, $\rho$ is a stationary state of $\cL_A$ whenever $\rho$ is a stationary state of $\cL$. 
\end{definition}

To conclude this subsection, we recall the definition of the spectral gap of a Liouvillian. 

\begin{definition}
Let $\cL :\cB_\Lambda\rightarrow\cB_\Lambda$ be a primitive reversible Liouvillian with stationary state $\rho$. The spectral gap of $\cL$ is given by 
\be \lambda_\Lambda(\cL):=\inf_{f\in\cA_\Lambda} \frac{-\avr{f,\cL_\Lambda(f)}_\rho}{\Var_\Lambda(f)}\label{eqn:vargap}\ee
\end{definition}

The significance of $\lambda_{\Lambda}$ follows from Theorem 2.2 of \cite{LogSobolev}, where it was shown that for every state $\sigma$,
\begin{equation}
\Vert e^{ t \cL_{\Lambda}^*}(\sigma) - \rho  \Vert_1 \leq \lambda_{\min}(\rho)^{-1/2} e^{- \lambda_{\Lambda} t},
\end{equation}
with $\rho$ the fixed point of $\cL_{\Lambda}^*$, $\lambda_{\min}(\rho)$ its minimum eigenvalue, and $\Vert X \Vert_1$ the trace-norm of $X$. In this paper $\rho$ will always be the thermal state of a local Hamiltonian, in which case $\lambda_{\min}(\rho) = e^{O(|\Lambda|)}$, with $|\Lambda|$ the number of sites of the lattice $\Lambda$. Thus $e^{ t \cL_{\Lambda}^*}(\sigma)$ to a good approximation of its fixed point in time of order $|\Lambda| / \lambda$. 

In terms of spectral theory, the spectral gap of a primitive reversible Liouvillian is given by the smallest non-zero eigenvalue of $-\cL$ (in the $\bL_2$ space associated to $\rho$). In Section \ref{sec:main} we will introduce a generalization of the spectral gap to subsets of the lattice.

%%%%%%%%%%%%%%%%%%%%%%%%%%%%%%%%%%%%%%%%%%%%%%%%%%%%%%%%%%%%%%%%%%%%%%%%%%%%%%%%%%%%%%%%%%%%%%

\section{Conditional expectations}\label{sec:condexp}

In this section we introduce a set of maps called \textit{conditional expectations}, which we denote suggestively by $\bE$.  These maps will later on play the role of local quasi-projectors onto the Gibbs state.  In Refs. \cite{ZiggyMaj1,ZiggyMaj2,ZiggyMaj3} one variant of conditional expectations was studied in very much the same context as we do here, where ergodic properties of Gibbs samplers were the main focus.  Also, Petz considered a similar set of maps is the context of \textit{corse graining} operations \cite{Petz}.   

\begin{definition}[conditional expectations]\label{propertiesCE}
Let $\Lambda$ be a finite subset of $\bZ^d$ and $\rho\in\cS^+_\Lambda$ be a full rank state. We call $\bE : \cB_\Lambda \rightarrow \cB_\Lambda$ a conditional expectation of $\rho$ if it satisfies the following properties:
\begin{enumerate}
\item (complete positivity) $\bE$ is completely positive and unital. 
\item (consistency) For any $f\in \cA_\Lambda$, $\tr{\rho \bE(f)}=\tr{\rho f }$.
\item (reversibilty) For any $f,g \in \cA_\Lambda$, $\avr{\bE(f),g}_\rho=\avr{f,\bE(g)}_\rho$.
\item (monotonicity) For any $f\in \cA_\Lambda$ and $n\in\bN$, $\avr{\bE^n(f),f}_\rho\geq\avr{\bE^{n+1}(f),f}_\rho$.

\end{enumerate}
\end{definition}

The consistency condition is reminiscent of the classical conditional expectation, while the reversibility condition can be understood as a form of detailed balance with respect to the state $\rho$. The role of monotonicity is not a priori clear, but will turn out to be necessary for the applications which we have in mind. Sometimes, but not always, we will consider conditional expectations that are also projectors. In that case, the monotonicity condition holds with equality. 

We will describe two examples of conditional expectations which are especially useful in the context of Gibbs samplers. As we will see below and in  Section \ref{sec:gibbsSS}, in addition to satisfying properties $1-4$ above, they inherit locality properties from the lattice Hamiltonian and Liouvillian used to define them. 

\subsection{Local Liouvillian projectors}

Let $\Lambda\subseteq\bZ^d$ and consider a local primitive Liouvillian $\cL _\Lambda=\sum_{Z\cap\Lambda\neq0} \cL_Z$ with stationary state $\rho\in\cS^+_\Lambda$. The \textit{local Liouvillian projector} associated with $\cL$ on $A$ is given by
\be \bE^{\cL}_A(f):=\lim_{t\rightarrow \infty} e^{t \cL_A}\label{LLproj}\ee

Note that if $\cL$ is locally primitive then $\bE^{\cL}_A(f)$ acts non-trivially only on $A^c$. If $\cL$ is frustration free, then $\bE^{\cL}_A$ is a conditional expectation with respect to the stationary state of $\cL$. Indeed:

 \begin{proposition}\label{thm:liouvilleprojector}
Let $\Lambda$ be a finite subset of $\bZ^d$. Let $\cL_\Lambda:\cB_\Lambda\rightarrow\cB_\Lambda$ be frustration free and locally reversible with respect to $\rho\in\cS_\Lambda$. Then for any $A\subseteq\Lambda$, $\bE_A^\cL$ is a conditional expectation with respect to $\rho$. 
\end{proposition}
\proof{
Complete positivity follows by construction, since for any $t\geq0$, $e^{t\cL}$ is a completely positive unital map. Consistency follows from frustration freeness. Indeed, assume $\rho$ is a stationary state of $\cL$, then by frustration freeness, for any $A\subseteq\Lambda$ and $f\in\cA_\Lambda$, 
\be \tr{\rho e^{t\cL_A}(f)}=\tr{e^{t \cL_A^*}(\rho)f}=\tr{\rho f}\nonumber\ee
Reversibility of $\bE_A^\cL$ follows directly from local reversibility of $\cL_\Lambda$. Finally, monotonicity can be seen to hold universally with equality from the projector property. For any $A\subseteq\Lambda$, note that 
\bea (\bE_A^\cL)^2(f)&=&\lim_{t\rightarrow \infty} e^{t\cL_A}(\lim_{t'\rightarrow\infty}e^{t'\cL_A}(f))\nonumber\\
&=& \lim_{t\rightarrow \infty} \lim_{t'\rightarrow\infty}e^{(t+t')\cL_A}(f)\nonumber\\
&=& \bE_A^\cL(f)\nonumber\eea
It immediately follows that $\avr{(\bE^\cL)^n(f),f}_\rho=\avr{(\bE^\cL)^{m}(f),f}_\rho$ for all $m,n\in\bN^+$.\qed}

It is clear from Eqn. (\ref{LLproj}) that if $\cL$ is local, then $\bE_A^\cL$ acts only on $A$ plus a finite neighbourhood around $A$ whose radius is upper bounded by the range of $\cL$. When $\cL_\Lambda$ is primitive, the expectation value with respect to the full system, namely $\rho: f \rightarrow \tr{\rho f}$, is equivalent to the local Liouvillian projector onto the whole system $\Lambda$.

\subsection{Minimal conditional expectations} \label{mincondexp}

The minimal conditional expectation $\bE_A^{\rho}$ is meant to minimally affect the observables outside of $A$ while still satisfying all four conditions of Definition \ref{propertiesCE}. This map has been considered previously, under the name \textit{corse graining map} in Ref. \cite{Petz} and \textit{block spin flip map} in Ref. \cite{ZiggyMaj1}. 

Let $\rho\in\cS^+_\Lambda$ be a full rank state on the lattice $\Lambda$ and let $A\subseteq\Lambda$. The \textit{minimal conditional expectation} of $\rho$ on $A$ is given by 

\be \bE^\rho_A(f) :=\Tr_{A}[\eta_A^\rho f \eta_A^{\rho \dag}],\label{Cond_Exp}\ee
where $\eta_A^\rho :=(\Tr_A[\rho])^{-1/2}\rho^{1/2}$. Recall that $\Tr_A$ is not the usual partial trace, but acts as a map from $\cB_\Lambda$ to itself given by the composition of the partial trace and tensoring with the identity matrix. A moment of thought shows that $\bE^\rho_A(f)$ is a hermitian operator  on the full system, which acts as the identity on subsystem $A$, and non-trivially on the rest of the system. We note that $\bE^\rho_A$ reduces to the classical conditional expectation of $\rho$ when the input observable is taken diagonal in the eigenbasis of $\rho$ \cite{ZiggyMaj1}. 

 \begin{proposition}\label{thm:localprojector}
Let $\Lambda$ be a finite subset of $\bZ^d$, $\rho\in\cS^+_\Lambda$, and $A\subseteq\Lambda$. Then $\bE^\rho_A$ is a conditional expectation with respect to $\rho$.
\end{proposition}

\proof{
Complete positivity follows directly from the explicit form of Eqn. (\ref{Cond_Exp}), as a composition of two completely positive maps. 
In order to show the other properties, we note the following useful fact about the partial trace. Denote $\rho_{A^c}\equiv \Tr_A[\rho]$, then 
\be\Tr_A[\rho_{A^c}^{-1/2}\rho^{1/2}f\rho^{1/2}\rho_{A^c}^{-1/2}]=\rho_{A^c}^{-1/2}\Tr_A[\rho^{1/2}f\rho^{1/2}]\rho_{A^c}^{-1/2}.\label{eqn:basicErho}\ee

 In particular, this implies that 
 \be\bE^\rho_A(\1)=\rho_{A^c}^{-1/2}\Tr_A[\rho]\rho_{A^c}^{-1/2}=\1,\ee 
 showing unitality of $\bE^\rho_A$. 
 
Consistency follows simply from Eqn. (\ref{eqn:basicErho}). Let $\Gamma_\rho(f):=\rho^{1/2}f\rho^{1/2}$. Then
\bea \tr{\rho\bE^\rho_A(f)}&=&\tr{\rho ~\Gamma_{\rho_{A^c}}^{-1}(\Tr_A[\Gamma_\rho(f)])}\nonumber \\
&=& \tr{\Tr_A[\Gamma_{\rho_{A^c}}^{-1}(\rho)]\Tr_A[\Gamma_\rho(f)]}\nonumber \\
&=& \tr{\Gamma_{\rho_{A^c}}^{-1}(\Tr_A[\rho] )\Tr_A[\Gamma_\rho(f)]}\nonumber \\
&=& \tr{\Tr_A[\Gamma_\rho(f)]}=\tr{\rho f},
 \eea

Reversibility follows by a similar argument:
\bea 
\avr{\bE^\rho_A(f),g}_\rho &=& \tr{\Tr_A[\Gamma_{\rho_{A^c}}^{-1}(\Gamma_\rho(f))]\Gamma_\rho(g)}\nonumber\\
&=& \tr{\Tr_A[\Gamma_\rho(f)]\Tr_A[\Gamma_{\rho_{A^c}}^{-1}(\Gamma_\rho(g))]}\nonumber\\
&=& \tr{\rho^{1/2}f\rho^{1/2}\Tr_A[\rho_{A^c}^{-1/2}\rho^{1/2}g\rho^{1/2}\rho_{A^c}^{-1/2}]} \nonumber \\ &=& \avr{f,\bE^\rho_A(g)}_\rho.
\eea

We now show monotonicity of $\bE^\rho$. Note that $\bE^\rho_A(\1)=\1$, hence the map is unital. Since it is also completely positive, this implies that its spectral radius is $1$. Furthermore, by reversibility, $\Gamma^{1/2}_\rho\bE^\rho_A(\cdot)\Gamma^{-1/2}_\rho$ is Hermitian (using the Hilbert-Schmidt inner product), so its spectrum is real and its left and right eigenvectors are the same. But since the spectrum of a matrix is unchanged by similarity transformations, we can write the spectral radius of $\bE^\rho_A$ as

\bea 1 &=& \sup_{f=f^\dag}\frac{\tr{\Gamma_\rho^{1/2}(f)\bE^\rho_A(\Gamma_\rho^{-1/2}(f))}}{\tr{f^2}}\\
&=& \sup_{g=g^\dag, f:=\Gamma_\rho^{1/2}(g)}\frac{\tr{\Gamma_\rho(g)\bE^\rho_A(g)}}{\tr{\Gamma_\rho(g)g}}\\
&=& \| \bE^\rho_A \|^2_{2\rightarrow 2,\rho},\eea
where the second line follows because  $\Gamma_\rho$ is hermicity preserving. In particular, this also implies that $\avr{f-\bE^\rho_A(f),f}_\rho\geq 0$ 

Now let $\tilde{\bE}_A(f)=\Gamma_\rho^{1/2}(\bE^\rho_A(\Gamma_\rho^{-1/2}(f)))$ and define $\Phi(f):= \tilde{\bE}_A^{1/2}(\Gamma_\rho^{1/2}(f))$. Then

\bea \sup_{f=f^\dag}\frac{\avr{\bE^\rho_A(f),\bE^\rho_A(f)}_{\rho}}{\avr{\bE^\rho_A(f),f}_{\rho}} &=& \sup_{f=f^\dag}\frac{\tr{\Gamma_\rho^{1/2}(f)\tilde{\bE}_A^2(\Gamma_\rho^{1/2}(f))}}{\tr{\Gamma_\rho^{1/2}(f)\tilde{\bE}_A(\Gamma_\rho^{1/2}(f))}}\\
&=& \sup_{f=f^\dag}\frac{\tr{\Phi(f) \tilde{\bE}_A(\Phi(f))}}{\tr{\Phi(f) \Phi(f)}}\\
&\leq& \sup_{g=g^\dag} \frac{\tr{g\tilde{\bE}_A(g)}}{\tr{f^2}}=1\eea

Thus, for all $f=f^\dag$, 

\be \avr{\bE^\rho_A(f),\bE^\rho_A(f)}_{\rho}\leq\avr{\bE^\rho_A(f),f}_{\rho}\ee

By iteration, this then shows monotonicity of $\bE^\rho$.

\qed}

\noindent \textit{Remark:}  $\bE^\rho_A$ is not a projector. However, if we take the limit of infinite iterations of the minimal conditional expectation of $\rho$ on $A\subseteq\Lambda$: $\lim_{n\rightarrow \infty}(\bE_A^\rho)^n$ then we recover a local projector satisfying the monotonicity condition with equality. The minimal conditional expectation has the benefit that it is uniquely defined for any full-rank state $\rho$. In other words, it does not invoke any dynamical description of the state $\rho$ (via a Liouvillian) as is the case for the local Liouville projector. On the other hand, it has the disadvantage that the map $\bE_A^\rho$ can potentially not exhibit any locality properties. We will see in the next section that in the special case when $\rho$ is the Gibbs state of a commuting Hamiltonian, then $\bE^\rho_A$ also acts on $A$ plus a neighbourhood around $A$, in the same way as $\bE_A^\cL$. 

%%%%%%%%%%%%%%%%%%%%%%%%%%%%%%%%%%%%%%%%%%%%%%%%%%%%%%%%%%%%%%%%%%%%%

\section{Gibbs states and Gibbs samplers}\label{sec:gibbsSS}

The primary purpose of this paper is to analyze the efficient preparation of Gibbs states of commuting Hamiltonians on finite dimensional lattices. In this section we introduce the notion of lattice Gibbs states in the quantum setting, and we describe two classes of Gibbs Samplers (Liouvillians) which prepare the Gibbs states of local commuting Hamiltonians. \\

Given a finite lattice $\Lambda\in\bZ^d$, let $\Phi :\Lambda\rightarrow \cA_\Lambda$ be an $r$-local bounded potential: i.e. for any $j\in\Lambda$, $\Phi(j)$ is a Hermitian matrix supported on a ball of radius $r$ around site $j$, and $||\Phi(j)||< K$ for some constant $K<\infty$. For any subset $A\subseteq\Lambda$, the Hamiltonian $H_A$ is given by 
\be H_A=\sum_{j\in A} \Phi(j).\ee
We say that $\Phi$ is a \textit{commuting potential} if for all $i,j\in\Lambda$, $[\Phi(i),\Phi(j)]=0$. 

Let $A\subseteq\Lambda$ and let $\Phi$ be a bounded local potential. Then we define the (outer) boundary of $A$ as:
\be \partial A :=\{j\in\Lambda| {\rm supp}(\Phi(j))\cap A\neq 0, j\notin A\}\ee
We denote by $A_\partial$ the union of $A$ and its boundary, i.e.  $A_\partial = A\cup \partial A$. Note that $H_A$ is supported on $A_\partial$.

The Gibbs (thermal) state of the full lattice $\Lambda$ is

\be \rho_\Lambda= e^{-\beta H_\Lambda}/\tr{e^{-\beta H_\Lambda}}\ee
Restricted Gibbs states will similarly be given by 
\be \rho_A=e^{-\beta H_A}/\tr{e^{-\beta H_A}},\label{eqn:gibbsstateA}\ee
for any $A\subseteq\Lambda$. Note that $\rho_A\in\cS_{A_\partial}^+$. 
Unless otherwise specified, $\rho$ will always be the Gibbs state of the full system. 

For classical spin systems, Gibbs states of local hamiltonians restricted to finite subsets of the square lattice can be unambiguously related to the Gibbs state in the thermodynamic limit, by parametrizing the effect of the ambient infinite system in the form of boundary conditions on the finite system. This procedure is often referred to as the DLR theory of boundary conditions. DLR theory shows that the contribution of the infinite ambient environment constitutes a convex set of boundary conditions \cite{israel}. Thus optimization over the set of boundary conditions can be restricted to particular pure configurations. This simple fact allows for remarkable simplifications when comparing properties of systems with varying lattice sizes, as is beautifully illustrated in Refs. \cite{MartinelliBook,MartinelliReview,Guionnet}. It turns out that these equivalences break down in the case of quantum systems (see Ref. \cite{WernerBC} for a detailed discussion and counter-examples). In this work, we circumvent DLR theory by working with conditional expectations of the Gibbs state. The price we pay is that our results are weaker than the analogous classical ones in many respects. 

We now turn to the description of Gibbs samplers of commuting Hamiltonians.

\subsection{Davies generators}\label{DaviesGenerator}

The dissipative dynamics resulting from the weak (or singular) coupling limit of a system to a large heat bath is often called \textit{Davies generator} \cite{Davies} or thermal Liouvillian \cite{LogSobolevCC}. 
We will not provide a full derivation here, but refer the reader to Ref.~\cite{LendlAlicki,Kristan,Davies} for a clear presentation and discussion of when this canonical form can be assumed.  

Given a finite subset $\Lambda$ of $\bZ^d$ and a local bounded potential $\Phi$, 
consider the system (self) Hamiltonian $H=\sum_{j\in\Lambda} \Phi(j)$, the bath (self) Hamiltonian $H_B$, and a set of system-bath interactions $\{S_{\alpha(k)}\otimes B_{\alpha(k)}\}$, where $\alpha(k)$ labels all the operators $S_{\alpha(k)}, B_{\alpha(k)}$ associated to site $k$. In this paper we will assume that the $\{ S_{\alpha(k)} \}_{\alpha(k)=1}^{ \text{dim}({\cal H}_k)^2 }$ are Hermitian and spans the algebra of observables in  ${\cal A}_{k}$. One concrete choice are the generalized Pauli matrices acting on site $k$.  

The total system-bath Hamiltonian is given by 
\be H_{tot}=H_\Lambda+H_B+ \sum_{\alpha(k), k \in\Lambda} S_{\alpha(k)}\otimes B_{\alpha(k)}\ee

Assuming the bath is in a Gibbs state, taking the coupling terms to zero, applying the Born-Markov approximation, and tracing out the bath yield the so-called Davies generators \cite{Davies}:

\be
	\cL_\Lambda^D(f) := i[H_\Lambda,f] +\sum_{k\in\Lambda}\cL^D_k(f),\label{eqn:davies1}\ee
where the local dissipative elements are given by

\be \cL^D_k(f) :=\sum_{\omega,\alpha(k)}\chi_{\alpha(k)}(\omega)\left(S^\dag_{\alpha(k)}(\omega)fS_{\alpha(k)}(\omega)-\half \{S^\dag_{\alpha(k)}(\omega)S_{\alpha(k)}(\omega),f\}_+\right),\nonumber
\ee
where $\omega$ are the so-called Bohr frequencies, and $\chi_{\alpha(k)}(\omega)$
are the Fourier coefficients of the two point correlation functions of the environment. The operators $S_{\alpha(k)}(\omega)$ are the Fourier coefficients of the system couplings $S_{\alpha(k)}$:
\be e^{-i t H_\Lambda}S_{\alpha(k)} e^{i t H_\Lambda}=\sum_\omega e^{i t \omega} S_{\alpha(k)}(\omega)\label{eqn:davies2}\ee
The $S_{\alpha(k)}(w)$
operators can be understood as mapping eigenvectors of $H_\Lambda$ with energy $\omega$ to eigenvectors of $H_\Lambda$ with energy $E+\omega$, and hence act in the Liouvillian picture as quantum jumps which transfer energy $\omega$ from the system to the bath and back. Reversibility of the map can be interpreted as the fact that the jumps to and from the system at a given energy are equally likely. 
The following useful relations hold for any $k\in\Lambda$, $\alpha(k)$ and $\omega$. Let $\rho$ be the Gibbs state of $H_\Lambda$, then for any $s\in[0,1]$,
\bea 	\chi_{\alpha(k)}(-\omega)&=&e^{-\beta \omega}\chi_{\alpha(k)}(\omega)\label{DBDavies1},\\ 
\rho^s S_{\alpha(k)}(\omega)&=&e^{s\beta \omega}S_{\alpha(k)}(\omega)\rho^s \label{DBDavies2}.
\eea 

We can naturally define a Liovillian given by the restriction of the Davies generator to the neighborhood of a subset of the lattice. For$A\subseteq\Lambda$:
\be \cL_A^D(f) := i[H_A,f]+\sum_{k\in A}\cL_k^D\nonumber\ee
We collect the properties of the Davies generator in the following lemma.

\begin{lemma}[properties of the Davies generator]\label{lem:DaviesGen}
Let $\Lambda$ be a finite subset of $\bZ^d$. Let $\Phi :\Lambda\mapsto \cA_\Lambda$ be an $r$-local bounded and commuting potential. Then the Davies generator $\cL^D$, as defined in Eqns. (\ref{eqn:davies1}-\ref{DBDavies2}) satisfies the following properties: \begin{enumerate}
\item For any $A\subseteq \Lambda$, $\cL^D_A$  is the generator of a completely positive unital semigroup $e^{t\cL^D_A}$. 
\item $\cL^D$ 
is $r$-local, meaning that each individual Lindblad term $\cL_k^D=\sum_{\omega,\alpha(k)}\cL^D_{\omega,\alpha(k)}$ acts non-trivially only on neighborhood of constant radius $r$ around $k\in\Lambda$. 
\item $\cL^D$ 
 is locally primitive and locally reversible with respect to the global Gibbs state $\rho$.
\item  $\cL^D$ 
is frustration free.
\end{enumerate}
\end{lemma}

\proof{

\vspace{0.05 cm}

\noindent 1. Complete positivity and unitality of $e^{t\cL^D_A}$ follow directly from the fact that  Eqn. (\ref{eqn:davies1}) is in Lindblad form, and by directly computing its action on $\id$. 

\vspace{0.1 cm}

\noindent 2. Locality of the Liouvillian follows from locality of the Lindblad operators $S_{\alpha(k)}(\omega)$. 
Given that for all $i,j\in\Lambda$, $[\Phi(i),\Phi(j)]=0$, we get for any $\alpha(k)$, 
\be e^{-i tH}S_{\alpha(k)}e^{i t H}=e^{-i t H_{k_\partial}}S_{\alpha(k)}e^{i t H_{k_\partial}},\nonumber\ee
which is manifestly local.

\vspace{0.1 cm}

\noindent 3. Local primitivity was shown to hold \cite{LendlAlicki} if for each $k\in\Lambda$, $\{S_{\alpha(k)}\}$ generates the full matrix algebra of site $k$. {Local reversibility} follows directly by exploiting the relations in Eqns. (\ref{DBDavies1}) and (\ref{DBDavies2}) to show that
$\avr{f,\cL_A^D(g)}_\rho=\avr{\cL_A^D(f),g}_\rho$, for any $A\subseteq\Lambda$.

\vspace{0.1 cm}

\noindent 4. Frustration freeness of the Davies generators is also implied by the local reversibility condition. Indeed, let $A\subseteq\Lambda$, then by local reversibility, for every $f,g\in\cA_\Lambda$, we get

\be \avr{f,\cL_A^D(g)}_\rho=\avr{\cL_A^D(f),g}_\rho\nonumber\ee
In particular, frustration freeness can be made explicit by choosing $f = \1$, obtaining
\be \tr{\rho\cL_A^D(g)}=\avr{\1,\cL_A^D(g)}_\rho= \avr{\cL_A^D(\1),g}_\rho=0,\nonumber\ee by local primitivity.
This implies that if the Liouvillian $\cL^D_\Lambda$ satisfies detailed balance with respect to the state $\rho$, then $\rho$ is a stationary state of $\cL^D_\Lambda$.
 \qed}

The Davies generators are often considered a good model for exploring thermalization in quantum systems. In particular, it is the standard approach for considering environment couplings in a variety of physical scenarios (e.g. for atomic or optical systems in the quantum regime).

\subsection{Heat-Bath generators}

We now consider a second class of Gibbs samplers. It has the drawback that it is less physically motivated. But has the advantage of simplicity.  

Let $\Lambda$ be a finite subset of $\bZ^d$. Let $\Phi :\Lambda\mapsto \cA_\Lambda$ be a local bounded and commuting potential, and let $\rho\in\cS^+_\Lambda$ be the associated Gibbs state. For some $A\subseteq\Lambda$, let $\bE^\rho_A$ be the minimum expectation value of $\rho$ on $A$ (defined in Section \ref{mincondexp}). We define the \textit{Heat bath} generator as 
\be \cL^H_A(f) :=\sum_{k\in A} (\bE^\rho_k(f)-f).\label{eqn:HeatBath}\ee 
Note that the conditional expectations are taken over single sites. Given that for any set of completely positive maps $\{T_j\}$, $\sum_j (T_j-\id)$ is a legitimate Liouvillian, we could have defined the heat bath Liouvillian with respect to essentially any set of conditional expectations. However, this specific choice is easier to work with, and closely mirrors the locality properties of the Davies generator.

The lemma below collects the relevant properties of the Heat-Bath generator:

\begin{lemma}[properties of the Heat-Bath generator]\label{lem:HBGen}
Let $\Lambda$ be a finite subset of $\bZ^d$. Let $\Phi :\Lambda\mapsto \cA_\Lambda$ be a $r$-local bounded and commuting potential. Then the associated heat-bath generator $\cL^H$ satisfies the following properties:
\begin{enumerate}
\item For any $A\subseteq \Lambda$, $\cL^H_A$  is the generator of a completely positive unital semigroup $e^{t\cL^H_A}$. 
\item $\cL^H$ 
is $r$-local, meaning that each individual Lindblad term $\cL_k^D$ acts non-trivially only on a neighbourhood with constant radius composed of sites around $k\in\Lambda$. 
\item $\cL^H$ 
 is locally primitive and locally reversible with respect to the global Gibbs state $\rho$.
\item  $\cL^H$ 
is frustration free.
\end{enumerate}
\end{lemma}

\proof{

\vspace{0.1 cm}

\noindent 1. Complete positivity and unitality of $e^{t\cL^H_A}$ follows directly from Eqn. (\ref{eqn:HeatBath}) and the fact that the minimal conditional expectation $\bE^\rho$ is completely positive and unital.

\vspace{0.1 cm}

\noindent 2. Locality can be seen by direct evaluation of one term in the generator. Let $k\in\Lambda$ and consider

\be \bE^\rho_k(f)=\Tr_{k}[\eta^\rho_k f \eta^{\rho\dag}_k]\nonumber\ee
where $\eta^\rho_k=(\Tr_{k}[e^{-\beta H_\Lambda}])^{-1/2}e^{-\beta H_\Lambda/2}$. 
But given that $\Phi$ is a commuting potential, we have
\bea (\Tr_{k}[e^{-\beta H_\Lambda}])^{-1/2}&=& (e^{-\beta H_{(k_\partial)^c}/2}\Tr_{k}[e^{-\beta H_{k_\partial}}]e^{-\beta H_{(k_\partial)^c}/2})^{-1/2}\nonumber\\
&=& e^{\beta H_{(k_\partial)^c}/4}(\Tr_{k}[e^{-\beta H_{k_\partial}}])^{-1/2}e^{\beta H_{(k_\partial)^c}/4}\nonumber\\
&=& (\Tr_{k}[e^{-\beta H_{k_\partial}}])^{-1/2}e^{\beta H_{(k_\partial)^c}/2},\nonumber\eea
with $H_{(k_\partial)^c}:=\sum_{j\neq k_\partial}\Phi(j)$. Here we have used that if two invertible hermitian operators $A,B$ commute $[A,B]=0$, then also $[A^{\alpha},B^{\beta}]=0$ for any scalars $\alpha,\beta\in\bR$, since two commuting Hermitian operators share the same orthonormal basis. 

Because the potential is commuting, $e^{-\beta H_\Lambda/2}=e^{-\beta H_{(k_\partial)^c}/2}e^{-\beta H_k/2}$. Thus we get

\be \eta^\rho_k = (\Tr_{k}[e^{-\beta H_{k_\partial}}])^{-1/2}e^{-\beta H_{k_\partial}/2}\nonumber\ee
Hence, if $\Phi$ is a commuting $r$-local potential, then $\cL_k^H$ is at most $r$-local, for any $k\in\Lambda$.

\vspace{0.1 cm}

\noindent 3. Local reversibility follows directly from reversibility of $\bE^\rho_A$ for any $A\subseteq\Lambda$. Local primitivity follows from Lemma \ref{HBprimitivity}. In order to prove that $\cL^H$ is locally primitive, we need to show that for any $A\subseteq\Lambda$,  $\cL^H_A(f)=0$ implies that $f$ acts non-trivially only on $A^c$. We show this by contradiction. Let $A\subseteq\Lambda$ and suppose that there exists an $g\in\cA_\Lambda$ with non-trivial support on $A$ such that  $\cL^H_A(g)=0$. Then it follows that $\avr{g,\cL^H_A(g)}_\rho=0$. From Lemma \ref{HBprimitivity} below, there exists a $C_A>0$ such that 
\bea \avr{g,-\cL^H_A(g)}_\rho &=& \sum_{k\in A}\avr{g,g-\bE^\rho_k(g)}_\rho\nonumber\\
&\geq& \frac{1}{C_A}\avr{g,g-\bE^\rho_A(g)}_\rho>0,\nonumber\eea
since $\bE_A^\rho(f)$ has support only on $A^c$, the support of $g$ has non-zero overlap with $A$.

\vspace{0.1 cm}

\noindent 4. Frustration freeness follows in much the same way as for the Davies generator, from the local reversibility of $\cL^H_\Lambda$ which is inherited from the reversibility of the minimal conditional expectation $\bE^\rho$.

\qed}

\noindent \textit{Remark:} The Heat-Bath generator has been considered previously in the context of lattice spin system in a series of papers \cite{ZiggyMaj1,ZiggyMaj2,ZiggyMaj3}. There the focus was on finding general local criteria for a quantum lattice system to be well defined in the thermodynamic limit. The results in Refs.   \cite{ZiggyMaj1,ZiggyMaj2,ZiggyMaj3} are hence  similar in spirit to ours, but quite different in scope and in terms of the methods used. Hence,  the two sets of results can be seen as being  complementary.

We conclude this section by pointing out that the whole framework of Heat-Bath Liouvillians also works if we replace the Gibbs state by some other full-rank (faithful) state of the lattice $\sigma$. It can be seen that Lemma  \ref{HBprimitivity} will still hold. However, it will typically be very difficult to obtain a bound on the locality of the individual terms. 

\begin{lemma}[equivalence of blocks  \cite{ZiggyMaj2}]\label{HBprimitivity}
Let $\Lambda$ be a finite lattice in $\bZ^d$. Let $\Phi :\Lambda\mapsto \cA_\Lambda$ be an r-local bounded and commuting potential, and let $\rho$ be the Gibbs state of $H_\Lambda=\sum_k \Phi(k)$. Let $A\subseteq\Lambda$, then there exist constants $c_A,C_A<\infty$ such that for any $f\in\cA_\Lambda$, the following inequality holds:
\be c_A \sum_{k\in A}\avr{f,f-\bE^\rho_k(f)}_\rho \leq\avr{f,f-\bE^\rho_A(f)}_\rho\leq C_A \sum_{k\in A}\avr{f,f-\bE^\rho_k(f)}_\rho.\ee
\end{lemma}

%%%%%%%%%%%%%%%%%%%%%%%%%%%%%%%%%%%%%%%%%%%%%%%%%%%%%%%%%%%%%%%%%%%%%

\section{Decay of Correlations}\label{sec:Clustering}

There are a number of ways of defining the correlations between observables in a quantum system. We will be interested in describing the situation when the correlations between local observables decay rapidly (exponentially) with the distance separating their supports. This behavior typically characterizes non-critical phases of many-body systems.

Let $A\subseteq\Lambda$ and \mjk{let $\bE\in\{\bE^\rho,\bE^\cL\}$ be one of the local conditional expectations of $\rho\in\cS^+_\Lambda$ introduced in Sec. \ref{sec:condexp}}. Then for any $f,g\in\cA_\Lambda$ we define the conditional covariance with respect to $\bE$ on $A$ by
\be \Cov_{A}(f,g) :=\left|\avr{f-\bE_A(f),g-\bE_A(g)}_\rho\right|,\ee 
and, similarly, the conditional variance $\Var_A(f) :=\Cov_A(f,f)$. We note that the conditional covariance with respect to the full lattice $\Lambda$ reduces to the usual covariance.

\begin{figure}
\centering
  \includegraphics[scale=0.30]{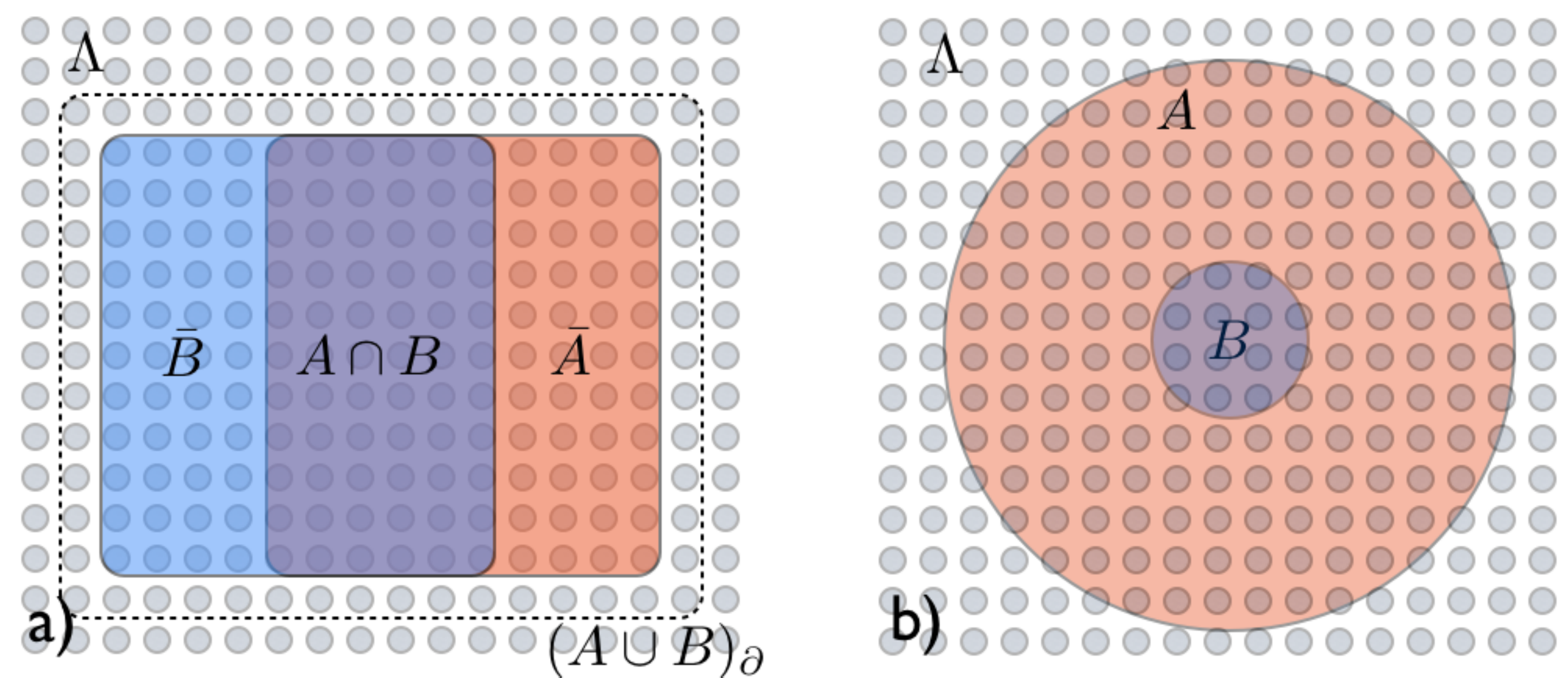}
    \caption{a) A subset $A\cup B\subseteq \Lambda$ of the full system, where $A\cap B\neq\emptyset$. We denote $\bar{A}\equiv (A\cup B)\setminus B$ and $\bar{B}\equiv (A\cup B)\setminus A$. The dotted lines around $A\cup B$ represent the boundary, which includes all terms of the Hamiltonian that overlap with $A\cup B$. The relevant distance is the width of the region $A\cup B$. b) $B\subseteq A \subseteq\Lambda$, where the relevant distance   is between the boundary of $B$ and the boundary of $A$. }
    \label{fig1}
\end{figure}

\begin{definition}[weak clustering]\label{def:weakclustering}
Let $\Lambda$ be a finite subset of $\bZ^d$ and let $\rho\in\cS^+_\Lambda$.  We say that $\rho$ satisfies {\rm weak clustering} if there exit constants $c,\xi>0$, such that for any observables $f,g\in\cA_\Lambda$, 
\be \Cov(f,g)\leq c ||f||_{2,\rho} ||g||_{2,\rho} e^{-d(\Sigma_f,\Sigma_g)/\xi},\label{weakclust}\ee
where $\Sigma_f$ ($\Sigma_g$) is the support of observable $f$ ($g$). 
\end{definition}

\begin{definition}[strong clustering]\label{def:strongclustering}
Let $\Lambda$ be a finite subset of $\bZ^d$ and let $\rho\in\cS^+_\Lambda$. \mjk{Let $\bE\in\{\bE^\rho,\bE^\cL\}$.} We say that $\rho$ satisfies {\rm strong clustering} with respect to $\bE$ if for any $A,B\subseteq\Lambda$ with $A\cap B\neq \emptyset$, there exist constants $c,\xi>0$ such that for any $f\in\cA_\Lambda$, 
\be \Cov_{A\cup B}(\bE_A(f),\bE_B(f))\leq c ||f||^2_{2,\rho} e^{-d(B\setminus A,A\setminus B)/\xi}.\label{strongclust}\ee
\end{definition}

It turns out that for the two conditional expectations considered in this paper it suffices to consider strong clustering for observables $f$ that act only on $A \cup B$ plus its boundary:

\begin{proposition}\label{prop:restrictionB}
Let $\Lambda$ be a finite subset of $\bZ^d$. Let $\Phi :\Lambda\mapsto \cA_\Lambda$ be an $r$-local bounded and commuting potential, and let $\rho$ be the Gibbs state of $H_\Lambda=\sum_k \Phi(k)$. \mjk{Let $\bE\in\{\bE^\rho,\bE^\cL\}$, then}
\begin{equation}\label{eqn:restB}
\sup_{f \in {\cal A}_{\Lambda}}  \frac{\Cov_{A\cup B}(\bE_A(f),\bE_B(f))}{||f||^2_{2,\rho}} = \sup_{f \in {\cal A}_{(A \cup B)_{\partial}}}  \frac{\Cov_{A\cup B}(\bE_A(f),\bE_B(f))}{||f||^2_{2,\rho_{A\cup B}}}.
\end{equation}
\end{proposition} 

\proof{
We consider the expression on the left hand side of Eqn. (\ref{eqn:restB}), and note that
\bea \sup_{f \in {\cal A}_{\Lambda}}  \frac{\Cov_{A\cup B}(\bE_A(f),\bE_B(f))}{||f||^2_{2,\rho}}  &=& \sup_{f \in {\cal A}_{\Lambda}}  \frac{\avr{f,(\bE_A-\bE_{A\cup B})(\bE_B-\bE_{A\cup B})(f)}_\rho}{\avr{f,f}_\rho}\nonumber\\
&=&\sup_{g\in\cA} \frac{\tr{g\hat{\cW}_{A\cup B}(g)}}{\tr{g^2}}\label{eqn:restB2}\eea 
where we made the replacement $g :=  \Gamma^{1/2}_{\rho}(f)$ and  $\Gamma_\rho (f):=\rho^{1/2}f\rho^{1/2}$ and defined the operators $\cW_{A\cup B}(f):=(\bE_A-\bE_{A\cup B})(\bE_B-\bE_{A\cup B})$ and $\tilde{\cW}_{A\cup B}=\Gamma^{1/2}_{\rho}\cW_{A\cup B} \Gamma^{-1/2}_{\rho}$. We note that $\tilde{\cW}_{A\cup B}$ is a hermitian operator, so Eqn. (\ref{eqn:restB2}) is simply an eigenvalue equation. 

Now, we will show that $\tilde{\cW}_{A\cup B}$ acts non-trivially only on $(A\cup B)_\partial$. We  write the subscript of $\rho$ explicitly so as to avoid confusion. 
Indeed, for any $g\in\cA_\Lambda$,
\bea \tilde{\cW}_{A\cup B}(g)&=& \Gamma_{\rho_\Lambda}^{1/2}\circ\cW_{A\cup B}\circ \Gamma_{\rho_\Lambda}^{-1/2} ( g) \\
&=& \Gamma_{\rho_{(A\cup B)_\partial}\rho_{((A\cup B)_\partial)^c}}^{1/2}\circ\cW_{A\cup B}\circ \Gamma_{\rho_{(A\cup B)_\partial}\rho_{((A\cup B)_\partial)^c}}^{-1/2} ( g) \\
&=& \Gamma_{\rho_{(A\cup B)_\partial}}^{1/2}\circ\cW_{A\cup B}\circ \Gamma_{\rho_{(A\cup B)_\partial}}^{-1/2} ( g).
\eea
Note that $\Gamma_{\rho_{(A\cup B)_\partial}}^{1/2}\circ\cW_{A\cup B}\circ \Gamma_{\rho_{(A\cup B)_\partial}}^{-1/2}$ only acts on $(A\cup B)_\partial$. But it is well known that the supremum in the variational characterization of the spectral radius is obtained by eigenvetors which have the same support as the operator. Thus, we recover the right-hand-side of Eqn. (\ref{eqn:restB}).
\qed}

\bigskip
\noindent \textit{Remarks:}

\begin{enumerate}[i]
\item The fact that the  observables can be restricted to having support only on the boundary of the region $A \cup B$ is reminiscent of the definition of clustering for classical spin systems, in which one is allowed to fix the boundary around a given region. Strong clustering goes in the direction of the Dobrushin-Schlossman uniqueness conditions \cite{DobrushinSchlossman}. However, since DLR theory of boundary conditions does not hold quantum mechanically \cite{WernerBC}, the stronger form of mixing in Eqn. (\ref{strongclust}) can not be expressed as local conditions which depend on individual boundary terms. This is because of the possibility of entangled boundary conditions, which have, as far as the authors know, not been studied much so far in the context of Gibbs states. 
\item Whenever $A\cup B=\Lambda$ then weak clustering also implies strong clustering since for any conditional expectation $\bE$, $\bE_A(f)$ has support on the complement of $A$, so that for any $h\in\cA_\Lambda$,
\bea \Cov_\Lambda(\bE_A(h),\bE_B(h))&\leq& c||\bE_A(h)||_{2,\rho}||\bE_B(h)||_{2,\rho} e^{-d(B\setminus A,A\setminus B)/\xi}\\
&\leq& c || h ||_{2,\rho} e^{-d(B\setminus A,A\setminus B)/\xi}\label{equiv3clust}\eea  
\end{enumerate}
\bigskip

We have defined exponential clustering in Eqns. (\ref{weakclust}) and (\ref{strongclust}) as being exponential decay of the covariance (with distance) weighted by the $\bL_2$ norms of the observables, instead of the operator norms, which is more common in the field. Given the ordering of $\bL_p$ norms (Lemma \ref{Lem:Lp-norm}), the operator norm always dominates the $\bL_2$ norm, so that our definitions of exponential clustering are strictly stronger than the usual ones. Although classically, the $\bL_2$ and the $\bL_\infty$ clustering are equivalent \cite{MartinelliReview}, we do not know if that is the case for quantum systems, even for commuting Hamiltonians. 

Note that the inner product used here is also unconventional. In particular, when $\rho$ is a pure state, then $\Cov_\rho(f,g)$ does not reduce to the usual pure-state covariance 
\be \left|\bra{\psi}f g\ket{\psi}-\bra{\psi}f \ket{\psi}\bra{\psi} g\ket{\psi}\right|,\label{eqn:classClust}\ee 
However, it is easy to define a modified (non-symmetric) covariance 
\be \Cov^{(0)}_\Lambda(f,g) := |\tr{\rho f^\dag g}-\tr{\rho f^\dag}\tr{\rho g}|,\ee
that reduces to Eq. (\ref{eqn:classClust}) when $\rho$ is pure. In general, $\Cov^{(0)}_\rho$ and $\Cov_\rho$ are not equivalent. But in the special case when $\rho$ is the Gibbs state of a commuting Hamiltonian, weak clustering in $\Cov^{(0)}_\rho$ implies weak clustering in $\Cov_\rho$. Indeed,

\begin{proposition}\label{equivgloClust}
Let $\Lambda$ be a finite subset of $\bZ^d$. Let $\Phi :\Lambda\mapsto \cA_\Lambda$ be an $r$-local bounded and commuting potential, and let $\rho$ be the Gibbs state of $H_\Lambda = \sum_{k \in \Lambda} \Phi(k)$. Then the following are equivalent:
\begin{itemize}
\item There exist constants $c_0,\xi_0>0$ such that
\be \Cov^{(0)}_{\Lambda}(f,g)\leq c_0 ||f||_{2,(0),\rho} ||g||_{2,(0),\rho} e^{-d(\Sigma_f,\Sigma_g)/\xi_0},\ee

\item There exist constants $c,\xi>0$ such that
\be \Cov_{\Lambda}(f,g)\leq c ||f||_{2,\rho} ||g||_{2,\rho} e^{-d(\Sigma_f,\Sigma_g)/\xi},\ee
\end{itemize}
where the modified $\bL_2$ norm is given by $||f||^2_{2,(0),\rho} :=\tr{\rho f^\dag f}$.
\end{proposition}
\proof{
Given an operator $f\in\cA_\Lambda$ with support on $A\subseteq\Lambda$, $\rho^s f \rho^{-s}$ has support on $A_\partial$, for any $s\in[0,1]$, since $\rho$ is the Gibbs state of a commuting Hamiltonian. Now define \mjk{$\tilde{f}\equiv \rho^{1/4}f\rho^{-1/4}$} and $\tilde{g}\equiv\rho^{1/4}g\rho^{-1/4}$. Then we get
\bea \Cov_{\Lambda}(f,g)&=& \left|\tr{\rho^{1/2}f^\dag\rho^{1/2}g}-\tr{\rho f^\dag}\tr{\rho g}\right|\\
&=& \left|\tr{\rho\tilde{f}^\dag\tilde{g}}-\tr{\rho \tilde{f}^\dag}\tr{\rho \tilde{g}}\right|\\
&=& \Cov_{\Lambda}^{(0)}(\tilde{f},\tilde{g})\\
&\leq& c \|\tilde{f}^\dag\|_{2,(0),\rho}\|\tilde{g}\|_{2,(0),\rho} e^{-(d(\Sigma_{f},\Sigma_{g})-2 r)/\xi}\\
&=&c \|f^\dag\|_{2,\rho}\|g\|_{2,\rho} e^{-(d(\Sigma_{f},\Sigma_{g})-2 r)/\xi},\eea
where we have used that 
\bea \|\tilde{g}\|^2_{2,(0),\rho}&=& \tr{\rho \tilde{g}^\dag \tilde{g}}\\
&=& \tr{\rho^{1/2}g^\dag\rho^{1/2}g}=\|g\|^2_{2,\rho}\eea
The other direction is identical, except that one has to define \mjk{$\tilde{f}\equiv \rho^{-1/4}f\rho^{1/4}$ and $\tilde{g}\equiv\rho^{-1/4}g\rho^{1/4}$}.
\qed}

\noindent \textit{Remark}:  A generalized covariance which interpolates between $\Cov^{(0)}_\rho$ and $\Cov_\rho$ was used in Ref.  \cite{Klietch}, and was shown to be a necessary ingredient in a proof of the existence of  a universal critical temperature above which correlations in the Gibbs state are exponentially clustering. This led to a stability theorem for locally perturbed Gibbs states at high temperatures. 

\subsection{Local indistinguishably}

One of the main contributions in this work is to introduce extended notions of clustering to characterize phases where correlations decay rapidly in a very strong sense. In this section, we consider how weak clustering relates to another important measure of correlation: 

 \begin{definition}[Local indistinguishability]\label{localindistinguishability}
Let $\Lambda$ be a finite subset of $\bZ^d$. Let $\Phi :\Lambda\mapsto \cA_\Lambda$ be an $r$-local bounded and commuting potential, and let $\rho$ be the Gibbs state of $H_\Lambda=\sum_{k \in \Lambda} \Phi(k)$. Let $B\subseteq A\subseteq\Lambda$ (see Fig. \ref{fig1}b), and \mjk{let $\bE\in\{\bE^\rho,\bE^\cL\}$}. Then, we say that \mjk{$\bE$ satisfies {\rm local indistinguishably}} if there exist constants $c,\xi>0$ such that for every two fixed points of $\bE^*_A$, $\rho^A$ and $\sigma^A$:
\be ||(\rho^A)_B-(\sigma^A)_B||_1\leq  c e^{-d(B,\Lambda\setminus A)/\xi}.\label{LTQO}\ee

\end{definition}

A condition similar to local indistinguishably was previously considered in Ref. \cite{Angelo}, and called \textit{Local Topological Quantum Order (LTQO)} because of an analogous condition for ground states of topologically ordered Hamiltonians (see also Refs. \cite{Michalakis,LTQOPEPS} for a closed system analogue). In the Gibbs sampler setting, as far as we know this condition does not appear to be connected with topological order, which is why we give it a different name here (see however the discussion in the outlook). 

We now show that a strengthening of weak clustering (changing  the bound from 2-norm to the product of 1- and infinity-norms) is equivalent to local indistinguishably:

\begin{theorem}\label{StrongCLusLTQO}
Let $\Lambda$ be a finite subset of $\bZ^d$. Let $\Phi :\Lambda\mapsto \cA_\Lambda$ be an $r$-local bounded and commuting potential,  $\rho$ be the Gibbs state of $H_\Lambda=\sum_{k \in \Lambda} \Phi(k)$, and \mjk{$\bE\in\{\bE^\rho,\bE^\cL\}$}. Then the following are equivalent:
\begin{itemize}
\item There exist constants $c_0,\xi_0>0$ such that for any $f,g\in\cA_\Lambda$, 
\be \Cov_{\Lambda}(f,g)\leq c ||f||_{1,\rho} ||g||_\infty e^{-d(\Sigma_f,\Sigma_g)/\xi},\label{strongclust1}\ee

\item  $\bE$ satisfies local indistinguishably. 
\end{itemize}

\end{theorem}

\proof{
We first show that Eqn. (\ref{strongclust1}) implies local indistinguishably. 
Let $\psi,\phi\in\cS_\Lambda$ be such that $\rho^A=\bE_A^*(\phi)$ and $\sigma^A=\bE_A^*(\psi)$. We now choose $A,C\subseteq \Lambda$ such that $B\cap C=\emptyset$ and $A\cup C=\Lambda$, $C=\Lambda\setminus B$ (as illustrated in Fig. \ref{fig1}b). Then

\bea ||(\rho^A)_B-(\sigma^A)_B||_1&=&\sup_{g\in\cA_B, ||g||=1} |\tr{\bE_A(g)(\phi-\psi)}|\nonumber\\
&=& \sup_{g\in\cA_B, ||g||=1} |\tr{(\bE_A(g)-\bE_{\Lambda}(g))\psi}-\tr{(\bE_A(g)-\bE_{\Lambda}(g))\phi}|\nonumber\\
&\leq& 2\sup_{g\in\cA_B, ||g||=1}|\tr{\phi(\bE_A(g)-\bE_\Lambda(g))}|\nonumber\eea

Now, defining $y:=\rho^{-1/2}\phi\rho^{-1/2}$, and noting that $\bE_C(g)=g$, we get

\bea ||(\rho^A)_B-(\sigma^A)_B||_1&=& 2\sup_{g\in\cA_B, ||g||=1} |\tr{\phi(\bE_A\bE_C(g)-\bE_\Lambda(g))}|\nonumber\\
&=& 2\sup_{g\in\cA_B, ||g||=1} |\avr{y,(\bE_A\bE_C-\bE_\Lambda)(g)}_\rho|\nonumber\\
&=& 2 \sup_{g\in\cA_B, ||g||=1} \Cov_\Lambda(\bE_A(y),\bE_C(g))\nonumber\\
&\leq& 2 \sup_{g\in\cA_B, ||g||=1} ||\bE_A(y)||_{1,\rho}||\bE_C(g)||_\infty e^{-d(B,\Lambda/A)/\xi}\nonumber\\
&=& 2 \sup_{g\in\cA_B, ||g||=1} ||y||_{1,\rho}||g||_\infty e^{-d(B,\Lambda\setminus A)/\xi}\nonumber\\
&\leq& 2 e^{-d(B,\Lambda/A)/\xi}\nonumber,\eea
where we have used that $||y||_{1,\rho}=\tr{\phi}=1$, and that $\bE_C(g)=g$. 

For the converse, let $A,B\subseteq\Lambda$ such that $A\cup B=\Lambda$ and let $f,g\in\cA_\Lambda$, where $f$ has support on $B^c$ and $g$ has support on $A^c$. Now consider

\bea \Cov_\Lambda(g,f)&=& \Cov_\Lambda(\bE_A(g),\bE_B(f))\nonumber\\
&=&|\avr{g,(\bE_A\bE_B-\bE_\Lambda)(f)}_\rho|\nonumber\\
&\leq& ||g||_{1,\rho}||(\bE_A\bE_B-\bE_\Lambda)(f)||_\infty\nonumber\\
&=&  ||g||_{1,\rho}\sup_{\varphi\in\cS_\Lambda} |\tr{\varphi (\bE_A\bE_B-\bE_\Lambda)(f)}|\nonumber\\
&=& ||g||_{1,\rho}\sup_{\varphi\in\cS_\Lambda} |\tr{(\bE_A^*-\bE_\Lambda^*)(\varphi)\bE_B(f)}|\nonumber\\
&=& ||g||_{1,\rho} ||\bE_B(f)||_\infty\sup_{\varphi\in\cS_\Lambda} ||\Tr_B[(\bE_A^*-\bE_\Lambda^*)(\varphi)]||_1\nonumber\\
&\leq&c||g||_{1,\rho} ||f||_\infty e^{-d(B\setminus A,A\setminus B)/\xi},
\eea
where we have used that $\bE_B$ is contractive in the $\infty \rightarrow\infty$ norm. 
\qed}

\noindent \textit{Remarks:}\\

\begin{enumerate}[i]
\item Eqn. (\ref{strongclust1}) only differs from the definition of weak clustering in that the norms on the right hand side are different. However, it is exactly this difference that allows the connection with local indistinguishably. Combining Theorem \ref{StrongCLusLTQO} and the results of Ref. \cite{Angelo}, we see that this form of clustering follows from having a system size-independent log-Sobolev constant, while weak clustering follows from the system having merely a constant spectral gap \cite{LogSobolevCC}. 
\item It is not know whether there is a relation between strong clustering and local indistinguishably. By analogy with classical results, one might expect that under certain conditions local indistinguishability implies strong clustering, but this is far from clear in our setting. It can be shown that local indistinguishability implies weak clustering, by Thm. \ref{StrongCLusLTQO} and an interpolation argument \footnote{This implication was pointed out to the authors by Angelo Lucia and David Perez-Garcia.}.
\end{enumerate}

%%%%%%%%%%%%%%%%%%%%%%%%%%%%%%%%%%%%%%%%%%%%%%%%%%%%%%%%%%%%%%%%%%%%%

\section{Main Results}\label{sec:main}

We are now in a position to prove the main results of the present work, namely  the equivalence between strong clustering of the Gibbs state and the associated Gibbs sampler (Heat Bath or Davies) being gapped. It turns out that both directions of the proof require different methods, hence for clarity of presentation we will separate them into two independent theorems. 

To start with, we recall the definition of the conditional variance, as it will play an important role in the proof. Let $\Lambda$ be the full lattice and let $A\in\Lambda$. Then the \textit{conditional variance} of $\rho\in\cS^+_\Lambda$ with respect to the conditional expectation $\bE$ on subset $A$ is given for any $f\in\cA_\Lambda$ by
\be \Var_A(f)=\avr{f-\bE_A(f),f-\bE_A(f)}_{\rho}=\|f-\bE_A(f)\|^2_{2,\rho}.\ee
The conditional variance reduces to the regular variance on the full lattice:
$\Var_\Lambda(f)=\avr{f,f}_{\rho}-\tr{\rho f }^2$. 

We now prove a proposition which relates the conditional variance of two subsets $A,B$ to the variance of their union $(A\cup B)$ when their overlap $A\cap B$ is non-zero (see Lemma 3.1 of \cite{Bertini} for a similar statement in the classical setting).  

\begin{proposition}\label{Variances}
Let $\Lambda$ be a finite subset of $\bZ^d$. Let $\rho\in\cS^+_\Lambda$ and let  $A,B\subseteq\Lambda$ be subsets of $\Lambda$ with non-zero overlap (i.e. $A\cap B\neq\emptyset$). Let \mjk{ $\bE\in\{\bE^\rho,\bE^\cL\}$}, and  suppose that there exists a positive constant $\epsilon>0$ such that for any $f\in \cA_\Lambda$, 
\be \Cov_{A\cup B}(\bE_A(f),\bE_B(f))\leq \epsilon \Var_{A\cup B}(f).\ee
Then
\be \Var_{A\cup B}(f)\leq (1-2\epsilon)^{-1}(\Var_A(f)+\Var_B(f))\ee
\end{proposition}
\proof{
Consider the following identity:
\bea 0&\leq& ||(\id-\bE_{A\cup B})\circ(\id-\bE_A-\bE_B)(f)||_{2,\rho}^2\nonumber\\
&=& -||(\id-\bE_{A\cup B})(f)||_{2,\rho}^2+||(\id-\bE_{A\cup B})\circ(\id-\bE_A)(f)||_{2,\rho}^2+\label{eq:varAB2}\\
&&||(\id-\bE_{A\cup B})\circ(\id-\bE_B)(f)||_{2,\rho}^2+2\avr{(\id-\bE_{A\cup B})\circ\bE_A(f),(\id-\bE_{A\cup B})\circ\bE_B(f)}_\rho\nonumber\eea

Recall that $||(\id-\bE_{A\cup B})(f)||_{2,\rho}^2=\Var_{A\cup B}(f)$,

\be ||(\id - \bE_{A\cup B})\circ(\id - \bE_A)(f)||_{2,\rho}^2\leq ||(f-\bE_A(f))||_{2,\rho}^2 = \Var_A(f)\nonumber\ee
and similarly for $\bE_B$, since $(\id-\bE_{A\cup B})$ is a positive contractive map,. Then we get from Eqn. (\ref{eq:varAB2}) that

\bea \Var_{A\cup B}(f)&\leq& \Var_A(f)+\Var_B(f)+2\avr{(\id-\bE_{A\cup B})\circ\bE_A(f),(\id-\bE_{A\cup B})\circ\bE_B(f)}_\rho\nonumber\\
&=&\Var_A(f)+\Var_B(f)+2\Cov_{A\cup B}(\bE_A(f),\bE_B(f))\nonumber\\
&\leq& \Var_A(f)+\Var_B(f)+2\epsilon \Var_{A\cup B}(f).\eea

This leads to the desired inequality

\be \Var_{A\cup B}(f)\leq (1-2\epsilon)^{-1}(\Var_A(f)+\Var_B(f)).\ee

\qed}

Note that for the proof of Proposition \ref{Variances} we have only used very general properties of the conditional expectations; in particular we have not assumed that $\rho$ is a Gibbs state, nor that $\bE$ has any local structure.

\subsection{Strong clustering implies gapped Gibbs sampler}

We now prove the first main theorem of the paper, which states that if the Gibbs state $\rho$ of a local commuting Hamiltonian satisfies strong clustering with respect to any of the two conditional expectations defined in Sec.\ref{sec:condexp}, then the associated Gibbs sampler (Heat Bath or Davies) has a spectral gap which is independent of the size of the lattice $|\Lambda|$.

By construction, the pairs $(\bE^\cL,\cL^D_\Lambda)$ and $(\bE^\rho,\cL^H_\Lambda)$ have the same kernel and share the same essential properties, such as reversibility and locality. In the remainder of this section, we will explicitly consider the Davies generator and associated Liouvillian expectation $(\cL^D,\bE^\cL)$, but it should be clear that the proof carries through almost unchanged for the pair $(\cL^H,\bE^\rho)$.

As the proof is based on an iterative construction comparing gaps of different sub lattices, it is necessary to define the concept of the gap of $\cL$ restricted to region $A$:

\begin{definition}\label{def:gap}
Let $\Lambda$ be a subset of $\bZ^d$. Let $\Phi :\Lambda\mapsto \cA_\Lambda$ be an $r$-local bounded and commuting potential, and let $\rho$ be the Gibbs state of $H = \sum_{k \in \Lambda} \Phi(k)$. The {\rm conditional gap} of $\cL^D_\Lambda$ with respect to $A\subseteq \Lambda$ is

\be \lambda_\Lambda(A) := \inf_{f\in\cA_\Lambda}\frac{\avr{f,-\cL^D_A(f)}_\rho}{\Var_A(f)}.\label{eqn:gap}\ee

\end{definition}
Given that $\bE^\cL$ is a projector and $\cL^D_A$ is reversible, Eq. (\ref{eqn:gap}) is identical to the spectral definition of the gap, i.e. the smallest non-zero eigenvalue of $-\cL^D_A$. In the case of $\cL^H_A$, this will not be so, since $\bE^\rho$ is not a projective conditional expectation. 

Note that the optimization is taken over operators with support on the full lattice $\Lambda$ with respect to the Gibbs state of the full lattice.  The gap $\lambda_\Lambda(A)$ is non-zero for any subset $A\subseteq\Lambda$ because by assumption $\bE_A$ and $\cL_A$ have the same kernel, and $(\bE,\cL)$ are assumed to be locally primitive\footnote{Recall  that the entire framework only makes sense for primitive semigroups, since $\bL_p$ spaces rely on a full rank reference state. }.

The next lemma tells us that when the potential $\Phi_\Lambda$ is commuting,  the conditional gap of $\cL$ with respect to $A$ in $\Lambda$  is the same as the conditional gap of $\cL$ with respect to $A$ in $A_\partial$.

\begin{lemma}\label{lemma:gap}
Let $\Lambda$ be a finite subset of $\bZ^d$. Let $\Phi :\Lambda\rightarrow\cA_\Lambda$ be an $r$-local bounded and commuting potential. Let $A\subseteq\Lambda$ and let $\rho_A$ be the Gibbs state of $H_A=\sum_{k\in A}\Phi_\Lambda(k)$ on $\cS_{A_\partial}$ (see Eqn. (\ref{eqn:gibbsstateA})) and $\rho$  the Gibbs state of $H_\Lambda=\sum_{k\in\Lambda}\Phi_\Lambda(k)$ on $\cS_\Lambda$. Then

\be \lambda_\Lambda(A)=\lambda_{A_\partial}(A):= \inf_{f\in\cA_{A_\partial}}\frac{\avr{f,-\cL^D_A(f)}_{\rho_{A}}}{||f-\bE_A(f)||^2_{2,\rho_{A}}}\ee
\end{lemma}
The proof is provided in the appendix, as it very closely resembles that of Prop. \ref{prop:restrictionB}.

\bigskip
\noindent \textit{Remark:} Lemma \ref{lemma:gap} tells us that even though the gap of $\cL_A$ is defined with respect to the full space $\cA_\Lambda$, the variational optimization reaches its maximum for observables that have support on $A_\partial$. This independence on the complement of $A$ is very important in the proof of the main theorem.

We are now in a position to state and prove our main theorem: 

\begin{theorem}\label{thm:main}
Let $\Lambda = (\bZ \backslash l)^d$, for an integer $l$. Let $\Phi :\Lambda\mapsto \cA_\Lambda$ be an $r$-local bounded and commuting potential, and let $\rho$ be the Gibbs state of $H_\Lambda=\sum_{k \in \Lambda} \Phi(k)$. There is a $l_0 > 0$ such that for every $l \geq l_0$, if $\rho$ satisfies strong clustering with respect to $\bE^\cL$, then $\cL_\Lambda^D$ has a spectral gap which is independent of $|\Lambda|$.
\end{theorem}

The proof strategy follows closely Ref. \cite{Bertini}, and consists in showing that for any sufficiently large subset $C\subseteq\Lambda$, we can choose $A,B\subseteq\Lambda$ such that $A\cup B=C$ and $A\cap B\neq\emptyset$ and such that the conditional gap with respect to $A$ or $B$ is approximately the same as the conditional gap with respect to  $C$. Choosing $A$ and $B$ to be roughly half the size of $C$, this shows that doubling the lattice size essentially leaves the conditional gap unchanged. By applying this procedure iteratively, we can show that the gap of the full system is lower bounded by the gap of a constant size subset. Finally, invoking Lemma \ref{lemma:gap}, we get that the conditional gap of a constant size subset of the lattice is lower bounded by a constant.  Hence the gap of $\cL$ on the full system cannot depend on the system size.   

\bigskip

\proof{

Let $A,B\subseteq\Lambda$ such that $A\cup B\neq \emptyset$ and $A\cap B$ forms a rectangle of minimal side  length $L$, while the overlap $A\cap B$ has minimum side length larger or equal to $\sqrt{L}$ (see Fig. \ref{fig1}). Recall that the strong clustering assumption implies there exist constants $c,\xi>0$ such that for any $f\in\cA_\Lambda$, 
\be \Cov_{A\cup B}(\bE_A(f),\bE_B(f))\leq c \Var_{A\cup B}(f) e^{-\sqrt{L}/\xi}\ee

Then, from Proposition \ref{Variances} and the definition of the conditional gap (Eqn. (\ref{eqn:gap})) we get that for any $f\in\cA_\Lambda$:
\bea \Var_{A\cup B}(f)&\leq&(1-2ce^{-\sqrt{L}/\xi})^{-1}(\Var_A(f) + \Var_B(f))\nonumber\\
&\leq& (1-2ce^{-\sqrt{L}/\xi})^{-1}\left(\frac{\avr{f,-\cL_A(f)}_\rho}{\lambda_\Lambda(A)}+\frac{\avr{f,-\cL_B(f)}_\rho}{\lambda_\Lambda(B)}\right)\\
&\leq& (1-2ce^{-\sqrt{L}/\xi})^{-1}\frac{1}{\lambda_\Lambda(A \wedge B)}(\avr{f,-\cL_{A\cup B}(f)}_\rho+\avr{f,-\cL_{A\cap B}(f)}_\rho),\nonumber\eea
where we have written $\lambda_\Lambda(A \wedge B):=\min\{ \lambda_\Lambda(A),\lambda_\Lambda(B)\}$.

At this point we might be tempted to upper bound $\avr{f,-\cL_{A\cap B}(f)}_\rho$ by $\avr{f,-\cL_{A\cup B}(f)}_\rho$. But this would provide us with a bound where the gap roughly halves in magnitude when we double the size of the system, which would lead upon iterations to a global gap decreasing polynomially with the system size. However, we can use an averaging trick originally developed for a similar purpose in classical lattice spin systems \cite{Bertini,MartinelliReview,MartinelliBook}. 

\begin{figure}
\centering
  \includegraphics[scale=0.30]{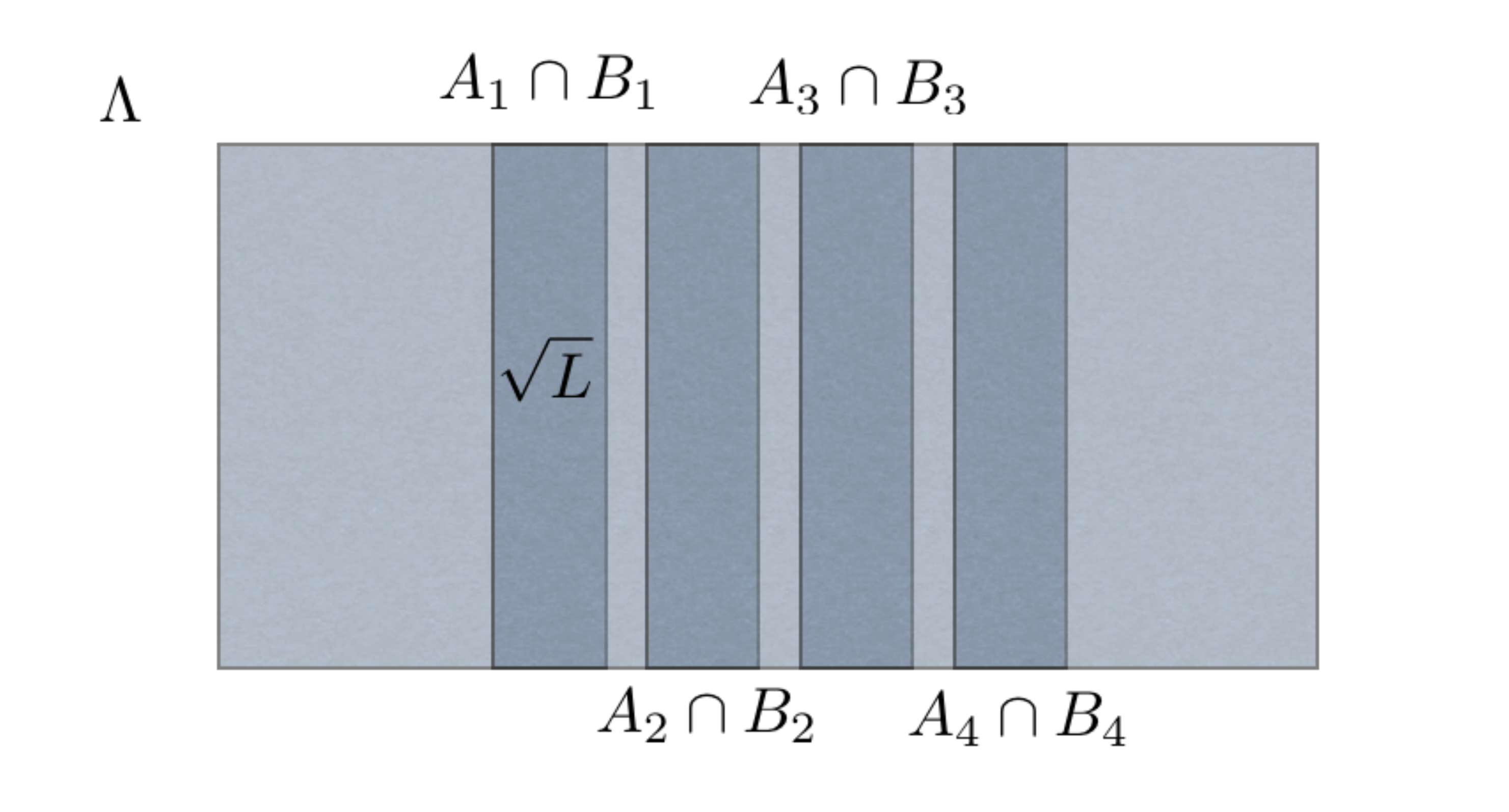}
    \caption{An illustration of the sequence of decompositions of the subset $C:=A_i\cup B_i$, which guarantee the conditions of Lemma \ref{lem:iteration}. The overlap between $A_j$ and $B_j$ is of order $\sqrt{L}$ for each $j$, and the different overlaps $A_j\cap B_j$ do not intersect.}
    \label{fig2}
\end{figure}

Given a rectangular subset of the lattice $C\subseteq\Lambda$, suppose there exists a sequence of subsets   $\{A_i,B_i\}_{i=1}^s$, where $s:=\lfloor L^{1/3}\rfloor$, with the property that for every $i=1,...,s$ 
\begin{itemize}
\item $A_i\cup B_i=C$,
\item $A_i\cap B_i$ has minimum side length lower bounded by $\sqrt{L}$,
\item  $A_i\cap B_i\cap A_j\cap B_j=\emptyset$ for all $i\neq j$.
\end{itemize}

Then, by noting that
\be \sum_{i=1}^s \avr{f,-\cL_{A_i\cap B_i}(f)}_\rho\leq\avr{f,-\cL_C}_\rho,\ee
we get 
\bea \Var_C(f)&\leq&(1-2ce^{\sqrt{L}/\xi})^{-1}\max_{i=1,..,s}\{(\frac{1}{\lambda_\Lambda(A_i)},\frac{1}{\lambda_\Lambda(B_i)})\}\nonumber\\
&&\left(\avr{f,-\cL_C(f)}_\rho+\frac{1}{s} \sum_{i=1}^s \avr{f,-\cL_{A_i\cap B_i}(f)}_\rho\right)\\
&\leq&(1-2ce^{\sqrt{L}/\xi})^{-1}(1+\frac{1}{s})\max_{i=1,..,s}\{(\frac{1}{\lambda_\Lambda(A_i)},\frac{1}{\lambda_\Lambda(B_i)})\}\avr{f,-\cL_C(f)}_\rho\nonumber
\eea

It is not difficult to see now that, for given $c,\xi$, there exists an $L_0$ such that for all $L\geq L_0$, 
\be (1-2ce^{\sqrt{L}/\xi})^{-1}\left(1+\frac{1}{\lfloor L^{1/3}\rfloor}\right)\leq \left(1+\frac{2}{\lfloor L^{1/3}\rfloor}\right),\ee 
which leads to 

\be \lambda_\Lambda(C) \geq \left(1+\frac{2}{\lfloor L^{1/3}\rfloor}\right)^{-1} \min_{i=1,..,s}\{\lambda_\Lambda(A_i),\lambda_\Lambda(B_i).\}\label{eqn:singlegapbd}\ee
Clearly, the specific value of $L_0$ depends on the constants $c,\xi$ an not on the system size $|\Lambda|$.

In order to complete the proof, we construct a decomposition of the full lattice $\Lambda$ into sequential subsets in such a way that we can use the bound in Eqn. (\ref{eqn:singlegapbd}) iteratively, and obtain a constant lower bound on $\lambda_\Lambda(\Lambda)$. The construction has been taken from Ref. \cite{Bertini} which in turn was initiated in the work of Martinelli \cite{MartinelliReview}.  
Let $l_k:=(3/2)^{k/d}$, and let $\cR_k^d$ be the set of all rectangles in $\cR^d$ which, modulo translations and permutations of the coordinates, are contained in 
\be [0,l_{k+1}]\times[0,l_{k+2}]\times ...\times[0,l_{k+d}]\ee
Assume that $\Lambda\equiv\cR^d_{k_{max}}$. We will later show that for any size $\Lambda$, the gap is always lower bounded. Note that we never explicitly compare the gaps of two systems $\Lambda_1$ and $\Lambda_2$ where each has specified boundary conditions, rather some boundary conditions are fixed for $\Lambda\equiv\cR^d_{k_{max}}$ and left untouched thereafter. 
We will also define the minimum gap restricted to rectangles in $\cR_k^d$ as $g_k:=\inf_{V\in \cR_k^d} \lambda_\Lambda(V)$. The idea behind this construction is that each rectangle in $\cR_k^d/\cR_{k-1}^d$ can be obtained as a slightly overlapping union of two rectangles in $\cR_{k-1}^d$. By Lemma \ref{lem:iteration}, we can then get the following iterative bound:

For all $C_k\subseteq \cR_{k+1}^d/\cR_{k}^d$ and $C_{k-1}\subseteq \cR_k^d/\cR_{k-1}^d$,

\be \lambda_\Lambda(C_k) \geq (1+2\frac{1}{\lfloor l_k^{1/3}\rfloor})^{-1} \lambda_\Lambda(C_{k-1})\ee

In particular the minimum $L_0$ can be associated with a minimum integer $k_0$, such that taking $\Lambda$ to be the thermodynamic lattice, and taking $k=k_0+1,...,\infty$, we get 

\bea \lim_{\Lambda\rightarrow\bZ^d}\lambda_\Lambda(\Lambda) &\geq& \left(\prod_{k=k_0+1}^\infty (1+2(3/2)^{-k/(3d)})\right)^{-1}\inf_{C_{k_0}\subseteq\cR_{k_0}}\lambda_\Lambda(C_{k_0})\\
&\geq&\exp[-2(1-(2/3)^{1/(3d)})]\inf_{C_{k_0}\subseteq\cR_{k_0}}\lambda_\Lambda(C_{k_0})\eea

Finally, since $k_0$ is determined only by $L_0$, it is independent of $|\Lambda|$. Hence, we get from Lemma \ref{lemma:gap} that $\lambda_\Lambda(C_{k_0})$ can be lower bounded by a constant independent of $|\Lambda|$ for any $C_{k_0}\subseteq\cR_{k_0}$. \qed}

\bigskip

\noindent \textit{Remarks: }

\begin{enumerate}[i]
\item In order for the iterative procedure to work, we do not need strong clustering to be exponential. Any polynomial decay with sufficiently high degree will in fact do the job. Furthermore, Theorem \ref{thm:main} actually shows that for every rectangle $A\subseteq\Lambda$, $\cL_A$ is gapped. 

\item If one were able to extend Prop. \ref{Variances} to show that 
\be \Ent_{A\cup B}(f)\leq (1-2\epsilon)^{-1}(\Ent_A(f)+\Ent_B(f)),\ee
where the $\Ent(f)$ was defined in Ref. \cite{LogSobolev}, and $\Ent_A(f)$ is some appropriately chosen conditional entropy, then the same proof strategy would work to show that the Gibbs samplers satisfy a Log-Sobolev inequality. This, in turn would have strong implications for the mixing of the semigroup, and would for instance show that strong clustering implies local indistinguishability, among other things. 

\item The proof of Theorem \ref{thm:main} requires strong clustering in order to invoke Proposition \ref{Variances}. If we were able to find a systematic procedure to associate boundary conditions $\zeta$ to the Hamiltonian $H_A$ for any $A\subseteq\Lambda$, and prove that $\lambda_\Lambda(A)\geq\lambda^\zeta_{A_\partial}(A_\partial)$, then it would be possible to prove Theorem \ref{thm:main} with respect to weak clustering rather than strong clustering. Indeed, whenever $A\cup B=\Lambda$ then Proposition \ref{Variances} holds under weak clustering. However, there is evidence to believe that Theorem \ref{thm:main} does not hold in general dimensions under the weak clustering assumption. Consider the 4D toric code, whose Davies generators are known not to be gapped (in fact the gap decreases exponentially). We know that the ground state of the 4D toric code satisfies a strong form of LTQO, where the suppression at large distances is not exponentially suppressed, but exactly zero. It is then plausible (although not proven) that the 4D toric code at finite temperature satisfies local indistinguishability. If true, this in turn would imply weak clustering at low enough temperature. Hence the equivalence between weak and strong clustering for the 4D toric Hamiltonian would to lead to a contradiction in general. See the outlook for a further discussion 

\item It is worth noting that in a sequence of important papers, Majewski and Zegarlinski have considered a similar approach to generalizing Glauber dynamics analysis to the quantum setting \cite{ZiggyMaj1,ZiggyMaj2,ZiggyMaj3}. There, they introduce the equivalent of our Heat Bath sampler, and show that the dynamics are well defined in the thermodynamic limit. Furthermore, they show that under some strong local ergodicity  conditions reminiscent of a certain form of the Dobrushin-Schlossmann complete analytic conditions, the dynamics are rapidly mixing, and in particular are gapped. Their conditions allow to show both our strong clustering, and local indistinguishability, and we hence expect them to be overly stringent for our main theorem.  Those conditions, in particular, also lead to a proof that the Heat Bath sampler is always gapped at high enough temperatures. 
\end{enumerate}

\begin{lemma}[\cite{Bertini}]\label{lem:iteration}
For all $C\subseteq \cR_k^d/\cR_{k-1}^d$ there exists a finite sequence $\{A_i,B_i\}_{i=1}^{s_k}$, where $s_k:=\lfloor l_k^{1/3}\rfloor$, such that 
\begin{enumerate}
\item $C=A_i\cup B_i$ and $A_i,B_i\in\cR_{k-1}^d$, for all $i=1,...,s_k$,
\item $d(C/A_i,C/B_i)\geq \frac{1}{8}\sqrt{l_k}$, for all $i=1,...,s_k$,
\item $A_i\cap B_i \cap A_j \cap B_j\emptyset$ if $i\neq j$.
\end{enumerate}
\end{lemma}
The proof is reproduced in the Appendix for the convenience of the reader.

%%%%%%%%%%%%%%%%%%%%%%%%%%%%%%%%%%%%%%%%%%%%%%%%

\subsection{Gapped Gibbs sampler implies strong clustering}\label{sec:converse}

We now proceed to prove the converse statement, namely that if a Gibbs state can be prepared by a gapped Gibbs sampler, then it satisfies strong clustering. The proof relies heavily on the so-called detectability lemma introduced in Refs. \cite{DetectLemma1,DetectLemma2}. 
We start by pointing out an important connection between Gibbs samplers of local commuting Hamiltonians, and general local frustration-free Hamiltonians. Indeed, let $\cL$ be a primitive Gibbs sampler of a local commuting Hamiltonian, and note that since it satisfy detailed balance, we get that the modified operator
\mjk{\be \hat{\cL}(f)= \rho^{1/4}\cL(\rho^{-1/4}f\rho^{-1/4})\rho^{1/4}\ee}
is Hermitian. In particular, if we represent $\hat{\cL}$ as a matrix on a doubled Hilbert space by the transformation $\ket{i}\bra{j}\mapsto\ket{i}\ket{j}$, we  get that $(-\hat{\cL})$ is a Hamiltonian (Hermitian operator) with ground state energy $0$. Throughout, we will use the same symbol for the super-operators acting on $\cB_\Lambda \rightarrow \cB_\Lambda$ and their associated operator representation acting on $\cB_\Lambda\otimes\bar{\cB}_\Lambda$. It should be obvious from the context which representation we are working with.  We furthermore have that $(-\hat{\cL})$ is local and frustration free (for both the Heat Bath and Davies Liouvillians). If $(-\hat{\cL})$ is gapped (in the Liouvillian sense) then $(-\hat{\cL})$ is also gapped (in the Hamiltonian sense).  The Gibbs state (density matrix $\rho$) is mapped onto a pure state $\ket{\sqrt{\rho}}=\sqrt{\rho}\otimes\1 \ket{\omega}$, where $\ket{\omega}=\sum_j\ket{jj}$ is proportional to the maximally entangled state, and satisfies 
\be \hat{\cL} \ket{\sqrt{\rho}}=0.\ee
Similarly, if $\bE$ is a projective conditional expectation, then $\hat{\bE}$ locally projects onto $\ket{\sqrt{\rho}}$.  We can summarize the correspondence as:
\begin{table}[h!]
\begin{center}

    \begin{tabular}{| l | l | l |}
    \hline
     & Commuting Gibbs Sampler & Frustration-free Hamiltonian \\ \hline
    State & Gibbs state $\rho$ & Ground state $\ket{\varphi}$\\ \hline
    Dynamics & Reversible Liouvillians $\cL$ & Hamiltonian $H$\\ \hline
    Projectors & Conditional Expectations $\bE$ & Ground state projectors $P$\\ \hline
    Gap & Spectral gap of $\cL$ &  Spectral gap of $H$ \\ \hline
    Framework & $\bL_p$ spaces & Hilbert spaces $\cH$ \\
    \hline
    \end{tabular}
     \caption{Correspondence between the Gibbs sampler framework and the Hamiltonian complexity framework.}
\end{center}
\end{table}

Thus, all of the tools developed for frustration-free Hamiltonians with a unique ground state can be applied to the setting of Gibbs samplers. In particular, we have completely recovered the setting of the detectability lemma and we can invoke the results from Ref. \cite{DetectLemma2}, as we now explain. 

Throughout the rest of this section we will be considering an $r$-local frustration-free Hamiltonian $-\hat{\cL}=\sum_{j\in\Lambda}-\hat{\cL}_j$, which has a system size ($\Lambda$) independent spectral gap $\lambda$ and ground state energy zero. It will in fact be more convenient to work with the modified Hamiltonian 
\be \hat{Q}_\Lambda=\sum_{j\in\Lambda} \hat{Q}_j\ee
where each local term $\hat{Q}_j:=(\1-\hat{\bE}^\cL_j)$ is a projector. It is not difficult to show (see e.g. \cite{DetectLemma2}) that if each $||\hat{\cL}_j||\leq K$ for some constant $K$, then the spectral gap $\epsilon$ of $Q$ is bounded as $\epsilon\geq \lambda/K$. Given that $\hat{Q}$ also has a unique ground state $\ket{\sqrt{\rho}}$, all results will also hold for $-\hat{\cL}$. 

Each term $\hat{Q}_j$ overlaps with a constant number of other local projectors $\hat{Q}_k$, so that the terms $\{ \hat{Q}_k\}$ can be partitioned into $g$ layers where each layer consists of non-overlapping projectors. See Fig. \ref{fig4} for an illustration of the one-dimensional case when $r=2$ and there are only two layers. Define 
\be \hat{\Pi}_j := \prod_{k\in {\rm layer (j)}} (\1-\hat{Q}_k)=\prod_{k\in {\rm layer (j)}} \hat{\bE}_k\ee
and 
\be\hat{\Pi} := \hat{\Pi}_g \cdots \hat{\Pi}_1.\ee 
 
Finally,  define $f(k,g)$ to be the number of sets of pyramids that are necessary to estimate the energy contribution of all the $Q_j$ terms. In the 1D case illustrated in Fig. \ref{fig4}, we had $f(k,g) = 2$. In the general case one can derive a crude upper bound:$f(g, k) \leq(g - 1)k^g$. For more details consult Ref.  \cite{DetectLemma2}. Then,
 
\begin{lemma}{[Detectability Lemma \cite{DetectLemma2}]}\label{lem:Detect}
With the notation introduced above, we get
\be || \hat{\Pi} - \hat{\bE}_\Lambda ||\leq \frac{1}{(\epsilon/f(k,g)+1)^{1/3}}.\ee
\end{lemma}

Since for all $j\in\Lambda$, $ \bE_\Lambda\bE_j=\bE_\Lambda$, we get that 
\be \lim_{n\rightarrow\infty}\hat{\Pi}^n=\hat{\bE}_\Lambda.\ee
Then from Lemma \ref{lem:Detect}, we know that there exists constants $C,\kappa>0$ such that 
\be ||\hat{\bE}_\Lambda - \hat{\Pi}^n||\leq C e^{-n \kappa},\label{eqn:deteclem3}\ee
where $\kappa$ is proportional to $\epsilon$. The exact same reasoning holds true for local projectors:
\be ||\hat{\bE}_A - \hat{\Pi}_A^n||\leq C e^{-n \kappa},\label{eqn:deteclem3}\ee
where $\Pi_A$ is constructed from local projectors that intersect $A$ (i.e. have support on $A_\partial$).

\begin{figure}
\centering
  \includegraphics[scale=0.28]{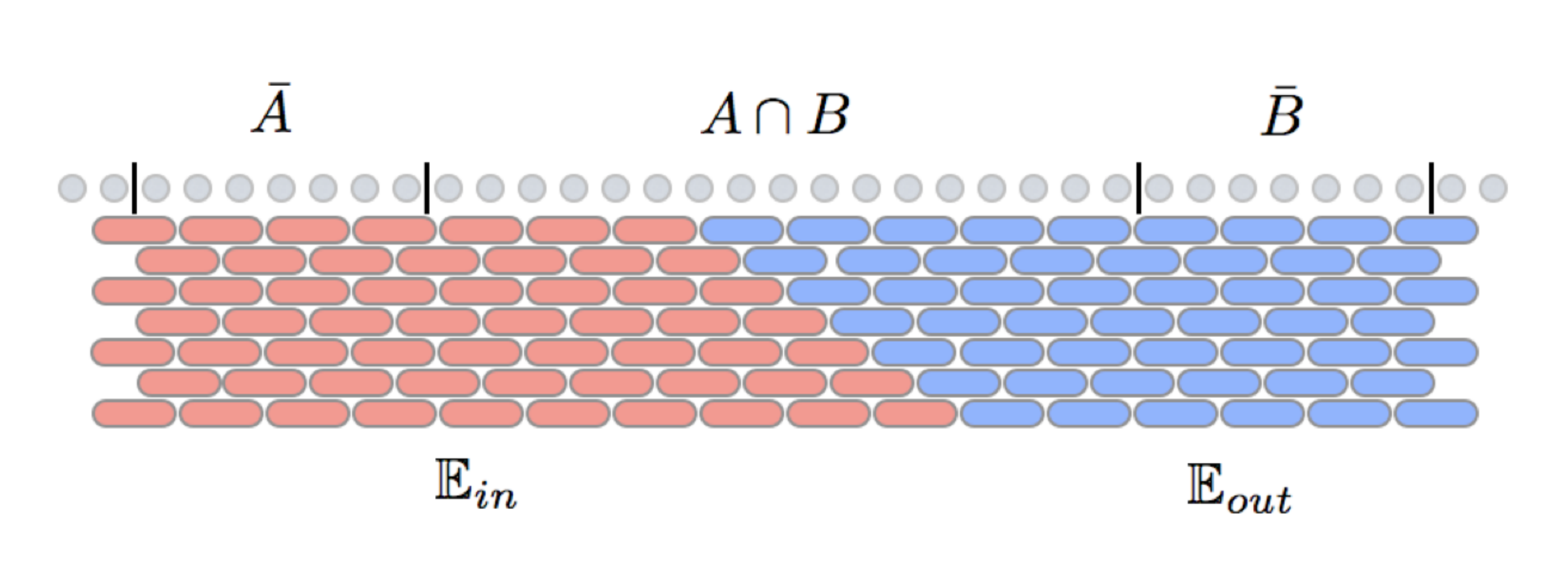}
    \caption{The setting of the detectability lemma when the Hamiltonian consists of $2$-local terms. We denote $\bar{A}\equiv (A\cup B)\setminus B$ and $\bar{B}\equiv (A\cup B)\setminus A$. The two non-overlapping approximate ground state projectors $\bE_{in}$ and $\bE_{out}$ are shown in red and blue.}
    \label{fig4}
\end{figure}

Now, we can use this approximate local projection property to prove the converse of Theorem \ref{thm:main}. 

\begin{theorem}\label{thm:converse}
Let $\Lambda$ be a finite subset of $\bZ^d$. Let $\Phi_\Lambda:\Lambda\mapsto \cA_\Lambda$ be an $r$-local bounded and commuting potential, and let $\rho$ be the Gibbs state of $H_\Lambda=\sum_{k \in \Lambda} \Phi(k)$. Denote $(\cL_\Lambda,\bE)$ a local Gibbs sampler of $\rho$, and suppose that $\bE$ is a projective conditional expectation. If for any $A\subseteq\Lambda$, $\cL_A$ is gapped, then $\rho$ satisfies strong clustering with respect to $\bE$. 
\end{theorem}
\proof{
Our proof resembles the proof of (weak) clustering for ground states of frustration free Hamiltonians in Ref. \cite{DetectLemma2} (see Section  6 of \cite{DetectLemma2}). Consider two subsets $A,B\subseteq\Lambda$ with $A\cap B\neq \emptyset$, and assume that the overlap has minimal side length $L$. Consider an approximate local projection 
\be\hat{\bE}_{A\cup B}\approx\hat{\Pi}^l_{A\cup B},\ee
where $\hat{\Pi}_{A\cup B}$ is restricted to local projective terms that intersect with the subset $A\cup B$. 
If $l\leq L/(gr)$, then we can write $\hat{\Pi}^l_{A\cup B}=\hat{\bE}_{in}\hat{\bE}_{out}$, where $\hat{\bE}_{in}$ consists of terms intersecting $\bar{A}$ on the first level, and the light cone resulting from the iterative application of $ \hat{\Pi}_{A\cup B}$. $\hat{\bE}_{out}$ is the rest of the local projective terms (see Figure \ref{fig4} for an illustration of the setting in 1D). 
    Let $f\in\cA_\Lambda$, then for $g_{\bar{A}}:=\bE_B(f)$ and $g_{\bar{B}}:=\bE_A(f)$, because of the frustration freeness and reversibility of the conditional projective expectations, we get 
\mjk{\bea \bE_{in}(g_{\bar{A}})&=&\bE_{in}\bE_B(g)=g_{\bar{A}}\label{eqn:DL1}\\
 \bE_{out}(g_{\bar{B}})&=&\bE_{out}\bE_A(g)=g_{\bar{B}}\label{eqn:DL2}\eea}

 Therefore, by the construction of $\Pi_{A\cup B}$, and Eqs. (\ref{eqn:DL1}) and (\ref{eqn:DL2}), we get 
 \bea\avr{g_{\bar{A}},\Pi^l_{A\cup B}(g_{\bar{B}})}_\rho&=&  \avr{g_{\bar{A}},\bE_{in}\circ\bE_{out}(g_{\bar{B}})}_\rho \nonumber\\
 &=& \avr{\bE_{in}(g_{\bar{A}}),\bE_{out}(g_{\bar{B}})}_\rho\\
 &=&\avr{g_{\bar{A}},g_{\bar{B}}}_\rho\nonumber\eea
 
 Then, noting that
 \be ||\Pi_{A\cup B}^l - \bE_{A\cup B}||_{2-2,\rho} =||\hat{\Pi}_{A\cup B}^l - \hat{\bE}_{A\cup B}||,\ee
 and using the $\bL_p$ H\" older's inequality, we get for any $f,h\in\cA_\Lambda$,
 
 \bea|\avr{f,(\Pi_{A\cup B}^l - \bE_{A\cup B})(h)}_\rho|&\leq& ||f||_{2,\rho}||(\Pi_{A\cup B}^l - \bE_{A\cup B})(h)||_{2,\rho}\\
 &\leq&||f||_{2,\rho}||\Pi_{A\cup B}^l - \bE_{A\cup B}||_{2-2,\rho}||h||_{2,\rho}.\eea
 Eq. (\ref{eqn:deteclem3}) leads to
 \bea \Cov_{A\cup B}(\bE_A(f),\bE_B(f))&=&\avr{g_{\bar{A}},g_{\bar{B}}}_\rho-\avr{g_{\bar{A}},\bE_{A\cup B}(g_{\bar{B}})}_\rho\nonumber\\
 &\leq&\avr{g_{\bar{A}},g_{\bar{B}}}_\rho-\avr{g_{\bar{A}},\Pi^l_{A\cup B}(g_{\bar{B}})}_\rho+C||g_{\bar{A}}||_{2,\rho}||g_{\bar{B}}||_{2,\rho}e^{-\kappa l}\nonumber\\
 &\leq& \avr{g_{\bar{A}},g_{\bar{B}}}_\rho-\avr{g_{\bar{A}},\bE_{in}\circ\bE_{out}(g_{\bar{B}})}_\rho +C||f||_{2,\rho}^2e^{-\kappa l}\nonumber\\
 &=& \avr{g_{\bar{A}},g_{\bar{B}}}_\rho-\avr{\bE_{in}(g_{\bar{A}}),\bE_{out}(g_{\bar{B}})}_\rho +C||f||_{2,\rho}^2e^{-\kappa l}\nonumber\\
 &=& C||f||_{2,\rho}^2e^{-\kappa l}.\eea
 
 \qed}

The proof of Thm. \ref{thm:converse} can be adapted to show that a gapped Gibbs sampler also implies weak clustering, which in turn shows that strong clustering implies weak clustering. \\

\begin{corollary}\label{cor:weakcluster}
Under the same assumptions as Theorem \ref{thm:converse}, 
\be \Cov(f,g)\leq c ||f||_{2,\rho}||g||_{2,\rho} e^{-d(\Sigma_f,\Sigma_g)/\xi}\ee
for some positive $c,\xi$. 
\end{corollary}
\proof{ The proof is identical to that of Theorem \ref{thm:converse}, but setting $A\cup B=\Lambda$, and taking $f$ and $g$ instead of $g_{\bar{A}}$ and $g_{\bar{B}}$ in the covariance.\qed}

\bigskip

\noindent \textit{Remarks:} 
\begin{enumerate}[i]
\item We point out that a number of results have already been published which show that a gapped Liouvillian implies weak clustering in the ground state \cite{Angelo,LogSobolevCC,PoulinLR}. Those results focus on general Liouvillians and their steady states. Hence although they are weaker than Theorem \ref{thm:converse}, they are more general.
 
\item Using the mapping described in Table 1, it can be seen that the strong clustering condition is essentially equivalent to condition $C3$ in Ref. \cite{Nachtergaele}. Theorem \ref{thm:main} can hence also be seen as an alternative proof of Nachtergaele's Thm. 3 in Ref. \cite{Nachtergaele}. It is furthermore interesting to note that condition  $C3$ in Ref. \cite{Nachtergaele} could in fact be related to a covariance decay condition; a connection which had thus far not been made.

\item  The detectablility lemma is almost sufficient to show local indistinguishability. Indeed, by H\" older duality,  one gets 
 \be|\avr{f,(\Pi_{A\cup B}^l - \bE_{A\cup B})(g)}_\rho|\leq ||f||_{1,\rho}||(\Pi_{A\cup B}^l - \bE_{A\cup B})(g)||_\infty.\ee
Thus, if one could show that $||(\Pi_{A\cup B}^l - \bE_{A\cup B})||_{\infty-\infty}$ is exponentially decaying in $l$, then local indistinguishability would follow. In the framework of frustration-free Hamiltonians and ground states, this would connect LTQO and the detectability lemma in an intriguing way, and could potentially lead to new strategies for proving the area law conjecture (which is implied by LTQO \cite{Michalakis}).

\item \mjk{The proof of the equivalence between strong clustering and a system size independent gap carries through with some modification  for the pair $(\cL^H,\bE^\rho)$. The only difference is that the  conditional expectation $\bE^\rho$ is not projective, and Theorem \ref{thm:converse} requires projective conditional expectation in order to use the detectability lemma. This restriction can be partially circumvented by noting that the gap of $\cL^H$ is lower bounded by the gap of $\cL^{H,\infty}(f):=\lim_{n\rightarrow\infty}\sum_{k\in\Lambda} ((\bE_k^\rho)^n(f)-f)$. Indeed, by monotonicity of conditional expectations, $\avr{f,\bE_k^\rho(f)}\geq\avr{f,(\bE_k^\rho(f))^n}$ for every $n\geq 1$. This implies that the gap of $\cL^H$ is larger than that of $\cL^{H,\infty}$ by the variational characterization of Eqn. (\ref{eqn:vargap}). Therefore, strong clustering in $\bE^\rho$ implies that $\cL^H$ is gapped, and gapped $\cL^H$ implies strong clustering in $\lim_{n\rightarrow\infty}(\bE^\rho)^n$. By the same argument, and Corollary \ref{cor:weakcluster}, strong clustering in $\bE^\rho$ implies weak clustering. }
\end{enumerate}

%%%%%%%%%%%%%%%%%%%%%%%%%%%%%%%%%%%%%%%%%%%%%%%%%%%%%%%%%%%%%%%%%%%%%

\section{One-dimensional models}\label{sec:1D}

In this section, we consider the special case of one-dimensional lattice systems. We show that the Gibbs sampler (Davies or Heat-Bath) is always gapped and hence satisfies strong clustering. This can be considered as a partial extension of the seminal result by Araki that Gibbs states of one dimensional lattice systems always satisfy clustering of correlations \cite{Araki}. The other result of this section is a proof that strong clustering and weak clustering are equivalent for one-dimensional commuting potentials. Intuitively, this is true because strong clustering is in a sense a statement of clustering restricted to a subsystem, where the worst-case boundary conditions are taken into account. Given that in a one-dimensional lattice, the boundary has dimension zero, its contribution only provides a constant multiplicative factor in the clustering statement. 

\begin{figure}[h]
\centering
  \includegraphics[scale=0.40]{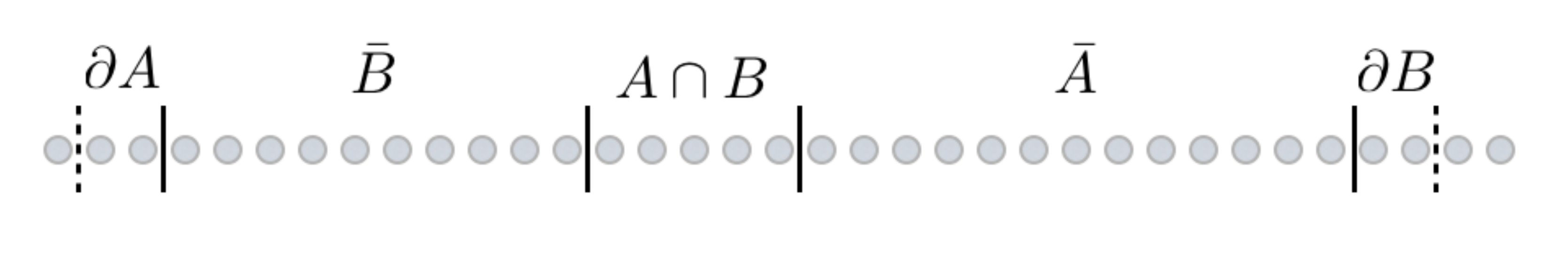}
    \caption{Illustration of a subset $A\cup B\subseteq \Lambda$ of a one-dimensional lattice. The boundary $\partial A \cup \partial B$ is zero-dimensional. }
    \label{fig5}
\end{figure}

\begin{theorem}[1D equivalence]\label{thm:1Dequiv}
Let $\Lambda = \bZ \backslash l$, for an integer $l$. Let $\Phi :\Lambda\mapsto \cA_\Lambda$ be an $r$-local bounded and commuting potential, and let $\rho$ be the Gibbs state of $H_\Lambda=\sum_{k \in \Lambda} \Phi(k)$. Then, $\rho$ satisfies weak clustering if, and only if, $\rho$ satisfies strong clustering for $\bE\in\{\bE^\rho,\bE^\cL\}$. 
\end{theorem}
\proof{ 
By Corollary \ref{cor:weakcluster}, strong clustering  implies weak clustering \mjk{(see Remark iv above)}. We show the converse below.

We can without loss of generality assume that $\bE_\Lambda(f)=0$. Given $A,B\subseteq\Lambda$, with $A\cap B\neq 0$, Proposition \ref{prop:restrictionB} shows that the maximum in 
\be \sup_{f\in\cA_\Lambda}\frac{\Cov_{A\cup B}(\bE_A(f),\bE_B(f))}{||f||_{2,\rho}^2}\leq c  e^{-d(\bar{A},\bar{B})/\xi},\label{1Dscluster}\ee
is reached for an operator $\1\otimes f$, with $f\in\cA_{(A\cup B)_\partial}$. Hence, it suffices to work with 
$f\in\cA_{(A\cup B)_\partial}$ and $\rho\propto e^{-\beta H_{A\cup B}}$ where $H_{A\cup B}=\sum_{k\in A\cup B}\Phi(k)$ (see also the comments after Lemma \ref{lemma:gap}).

Denote $h = \bE_A(\tilde{f})$ and $g = \bE_B(\tilde{f})$ where $\tilde{f}$ is the observable achieving the maximum on the LHS of Eq. (\ref{1Dscluster}), which we know acts only on $(A\cup B)_\partial\equiv \partial_A \cup A\cup B\cup \partial_B$ (see Fig. \ref{fig5}).
Throughout this proof, we will write the tensor products explicitly so as to avoid confusion. We furthermore define the modified states $\sigma_{A_\partial}\equiv \rho_A$, which for any $A\subseteq\Lambda$ is the Gibbs state restricted to a subset \textit{without} the Hamiltonian terms intersecting the boundary of that subset.
Since the Hamiltonian of the system is commuting,  we can write
\begin{equation}
\rho_{(A\cup B)_\partial} = P (\sigma_{\partial_A} \otimes \sigma_{A\cup B} \otimes \sigma_{\partial_B}) Q,
\end{equation}
where $P$ acts only on $\partial_A$ and $A$ (in fact only on a constant-sized region of $A$ that touches $\partial_A$) and $Q$ acts only on $\partial_B$ and $B$. 

Note furthermore that $h$ has support on the complement of $A$ in $(A\cup B)_\partial$, and similarly for $g$ with respect to $B$. 

We will prove that weak clustering is equivalent to strong clustering in the covariance $\Cov_\rho^{(0)}$. This is sufficient, because as shown in Proposition \ref{equivgloClust}, for commuting Hamiltonians, the two types of weak clustering are equivalent. We can write
\begin{equation} \label{inner}
\avr{h,g}^{(0)}_{\rho} \equiv \tr{\rho h g^\dag  } = \tr{ \sigma_{\partial_A} \otimes \sigma_{A\cup B} \otimes \sigma_{\partial_B} (Qh) (gP)   }
\end{equation}

Since $Qh$ has support on the complement of $A$, we can write it in its Schmidt decomposition with respect to the Hilbert spaces $(\partial_A\partial_B, \avr{\cdot,\cdot}^{(0)}_{\sigma_{\partial_A} \otimes \sigma_{\partial_B}})$ and $(B\setminus A, \avr{\cdot,\cdot}^{(0)}_{\sigma_{B\setminus A}})$ as
\begin{equation}
Q h= \sum_{k} h^{k}_{\partial_A \partial_B} \otimes h^k_{B\setminus A},
\end{equation}
where $h^{k}_{\partial_A \partial_B}\in\cA_{\partial_A \partial_B}$ and $h^k_{B\setminus A}\in\cA_{B\setminus A}$ are the Schmidt coefficients satisfying $\tr{\sigma_{\partial_A} \otimes \sigma_{\partial_B} (h^{k}_{\partial_A \partial_B})^\dag h^{k'}_{\partial_A \partial_B}} = \delta_{kk'} \tr{\sigma_{\partial_A} \otimes \sigma_{\partial_B} (h^{k}_{\partial_A \partial_B})^\dag h^{k}_{\partial_A \partial_B}}$, and likewise for $\{ h^k_{B} \}$. Note that $k$ varies from 1 to $(d_{\partial_A}d_{\partial_B})^2$, where $d_{\partial_A}$ is the dimension of the space spanned by the sites in $\partial_A$. 

Similarly, since $Pg $ has support on the complement of $A$, we can write it in its Schmidt decomposition with respect to the Hilbert spaces $(\partial_A\partial_B, \avr{\cdot,\cdot}^{(0)}_{\sigma_{\partial_A} \otimes \sigma_{\partial_B}})$ and $(A\setminus B, \avr{\cdot,\cdot}^{(0)}_{\sigma_{A\setminus B}})$ as
\begin{equation}
Pg= \sum_{k} g^{k}_{\partial_A \partial_B} \otimes g^k_{A\setminus B},
\end{equation}

From Eq. (\ref{inner}) we get
\begin{equation} 
\avr{h,g}^{(0)}_{\rho} = \sum_{k, l}  \tr{ (\sigma_{\partial_A} \otimes \sigma_{\partial_B}  h^{k}_{\partial_A \partial_B} (g^{l}_{\partial_A \partial_B})^\dag}  \tr{\sigma_{A\cup B} h^k_{B\setminus A} ( g^l_{A\setminus B})^\dag},
\end{equation}
and thus
\begin{eqnarray} 
|\avr{h,g}^{(0)}_{\rho}| &\leq& \sum_{k, l}  \left|\tr{ (\sigma_{\partial_A} \otimes \sigma_{\partial_B}  h^{k}_{\partial_A \partial_B} (g^{l}_{\partial_A \partial_B})^\dag}\right|  \left|\tr{\sigma_{A\cup B} h^k_{B\setminus A} ( g^l_{A\setminus B})^\dag}\right| \nonumber \\
&\leq&  \sum_{k, l} (\tr{ (\sigma_{\partial_A} \otimes \sigma_{\partial_B}  h^{k}_{\partial_A \partial_B} (h^{k}_{\partial_A \partial_B})^\dag})^{1/2} (\tr{ (\sigma_{\partial_A} \otimes \sigma_{\partial_B}  g^{l}_{\partial_A \partial_B} (g^{l}_{\partial_A \partial_B})^\dag})^{1/2} \nonumber\\
&& \left|\tr{\sigma_{A\cup B} h^k_{B\setminus A} ( g^l_{A\setminus B})^\dag}\right| \nonumber 
\end{eqnarray}

Noting that 
\be |\tr{\sigma_{A\cup B} g f^\dag}|\leq c |\tr{\rho_{A\cup B} g f^\dag}|,\ee
for some constant $c$ since $\rho_{A\cup B}$ and $\sigma_{A\cup B}$ only differ on the boundary of $A\cup B$, which is zero-dimensional. 

Then, by weak clustering, there exist constants $c,\xi>0$ such that
\begin{equation}
 \left|\tr{\sigma_{A\cup B} h^k_{B\setminus A} ( g^l_{A\setminus B})^\dag}\right| \leq c \Vert h^k_{B\setminus A} \Vert_{2,(0), \sigma'_{A\cup B} } \Vert g^l_{A\setminus B} \Vert_{2,(0), \sigma_{A\cup B}} e^{-l/\xi},
\end{equation}
and so
\begin{eqnarray} 
|\langle f, g \rangle^{(0)}_{\rho}| &\leq& c e^{-l/\xi} \sum_{k, l} \Vert h^{k}_{\partial_A \partial_B}\Vert_{2, (0),\sigma_{\partial_A} \otimes \sigma_{\partial_B} } \Vert g^{l}_{\partial_A \partial_B} \Vert_{2,(0), \sigma_{\partial_A} \otimes \sigma_{\partial_B}}  \nonumber\\
 &&\Vert h^k_{B\setminus A} \Vert_{2,(0), \sigma_{A\cup B} } \Vert g^l_{A\setminus B} \Vert_{2,(0), \sigma_{A\cup B}} \nonumber
\end{eqnarray}
By concavity of $x \mapsto x^{1/2}$,
\begin{eqnarray} 
|\langle f, g \rangle^{(0)}_{\rho}| &\leq& c e^{-l/\xi} d^2\left( \sum_{k, l} \Vert h^{k}_{\partial_A \partial_B}\Vert^2_{2, (0),\sigma_{\partial_A} \otimes \sigma_{\partial_B} } \Vert g^{l}_{\partial_A \partial_B} \Vert^2_{2,(0), \sigma_{\partial_A} \otimes \sigma_{\partial_B}} \right. \nonumber\\
 &&\left.\Vert h^k_{B\setminus A} \Vert^2_{2,(0), \rho_{A\cup B} } \Vert g^l_{A\setminus B} \Vert^2_{2,(0), \sigma_{A\cup B}} \right)^{1/2}\nonumber\\
 &=& c e^{-l/\xi} d^2 (\tr{\sigma' Qh (Qh)^\dag}\tr{\sigma' Pg (Pg)^\dag})^{1/2}
\end{eqnarray}
with $d := d_{\partial_A} d_{\partial_B}$ and $\sigma' := \sigma_{\partial_A} \otimes \sigma_{\partial_B} \otimes \sigma_{A\cup B}$. In the last line, we used that
\begin{eqnarray}
\tr{\sigma' Qh (Qh)^\dag} &=& \sum_{k, k'} \tr{ (\sigma_{\partial_A} \otimes \sigma_{\partial_B} \otimes \sigma_{A\cup B})  ( h^{k}_{\partial_A \partial_B} \otimes h^k_{B} ) (h^{k'}_{\partial_A \partial_B} \otimes h^{k'}_{B} )^{\cal y}  }\nonumber \\ 
&=&  \sum_{k} \tr{ (\sigma_{\partial_A} \otimes \sigma_{\partial_B} \otimes \sigma_{A\cup B})  ( h^{k}_{\partial_A \partial_B} \otimes h^k_{B} ) (h^{k}_{\partial_A \partial_B} \otimes h^k_{B} )^{\cal y}  }, \nonumber
\end{eqnarray}
which follows from the orthogonality of the $h^k$s. 

The result follows from the bound
\begin{equation}
\tr{\sigma' Qh (Qh)^\dag} \leq C(r) \Vert  f_{\partial_A (B\setminus A)\partial_B}  \Vert_{2,(0), \sigma}^2
\end{equation}
for a function $C(r)$ of the interaction range of the Hamiltonian.
\qed}

Using Theorem \ref{thm:1Dequiv} we can now show that commuting Gibbs samplers of one dimensional lattice systems are always gapped. At first sight the clustering proof of Araki \cite{Araki}, together with Thm. \ref{thm:1Dequiv}, Thm.  \ref{thm:main}, and Prop. \ref{equivgloClust}, should suffice to prove that the Davies generator of a 1D commuting Hamiltonian is gapped. However, Araki's result has the error term expressed in infinity norm, whereas we need an $\bL_2$ norm bound (see Eqn. (\ref{weakclust})). Here we use methods from the theory of matrix product states to show that the clustering results can in fact be recast in terms of $\bL_2$ norms.

\begin{proposition}
Let $\Lambda = \bZ \backslash l$, for an integer $l$. Let $\Phi :\Lambda\mapsto \cA_\Lambda$ be an $r$-local bounded and commuting potential, and let $\rho$ be the Gibbs state of $H_\Lambda=\sum_{k \in \Lambda} \Phi(k)$. Then, the Heat-Bath and Davies generators are gapped. 
\end{proposition}

\begin{proof}
Let $H$ be a commuting Hamiltonian in one dimension with finite range $r$, and let $\rho$ be its Gibbs state. Group the sites into blocks of $r/2$ sites, so that the Hamiltonian is 2-local.  Since the terms are pairwise commuting we can write the Gibbs state as 
\begin{equation}
\rho = \frac{1}{Z} \bigotimes_{i : \text{even}} e^{- \beta H_{i, i+1}}  \bigotimes_{i : \text{odd}} e^{- \beta H_{i, i+1}}. 
\end{equation}

Then it follows that 
\be \ket{\rho^{1/2}}=\rho^{1/2}\otimes\1 \ket{\omega}\ee 
is a matrix product state with bond dimension bounded by $2^r$, where $\ket{\omega}=\sum_j \ket{j,j}$ is proportional to the maximally entangled state. Indeed, for any bipartition $(1, \ldots i) (i+1, \ldots, n)$, only the term $e^{- \beta H_{i, i+1}}$ can increase the Schmidt rank. Thus $\ket{\rho^{1/2}}$ has Schmidt rank bounded by $2^r$ in every bipartite cut and by Ref. \cite{Vidal} it is a MPS of bond dimension $2^r$.

Now, Araki \cite{Araki} has shown that Gibbs states of one-dimensional bounded local Hamiltonians satisfy
\be 
\Cov^{(0)}(f,g)\leq c ||f|| ||g|| e^{- \text{d}(\Sigma_f, \Sigma_g)/\xi},
\ee
for some constants $c,\xi$, and $f,g$ with constant support. Then, invoking Prop. \ref{equivgloClust} and the transormation from Table I, we get that for every $f'$ and $g'$ that have the same supports on $f_{\partial}$ and $g_{\partial}$:

\begin{eqnarray}\label{MPSclust}
&& |\bra{\rho^{1/2}} f' \otimes g' \ket{\rho^{1/2}} - \bra{\rho^{1/2}} f' \ket{\rho^{1/2}} \bra{\rho^{1/2}} g' \ket{\rho^{1/2}} | \nonumber \\ &\leq&  c d^{|\Sigma_f|}  d^{|\Sigma_g|}  ||f'|| ||g'|| e^{- \text{d}(\Sigma_f, \Sigma_g)/\xi},
\end{eqnarray}
Therefore $\ket{\rho^{1/2}}$ has decay of correlations for every $f, g$ with constant supports.  

The parent Hamiltonian of a MPS has a unique ground state if, and only if, its transfer operator is gapped (with a unique top eigenvector).  Furthermore, the MPS satisfies Eq. (\ref{MPSclust}) (i.e. has decay of correlations for observables of constant size) if, and only if, the transfer operator is gapped \cite{FNW,PVWC07}. It follows that the transfer operator of $\ket{\rho^{1/2}}$ is gapped, which is equivalent to injectivity of the MPS. This in turn implies that the parent Hamiltonian of the MPS $\ket{\rho^{1/2}}$ is gapped \cite{Nachtergaele}.  Hence we can use the detectability lemma on the parent Hamiltonian of $\ket{\rho^{1/2}}$, and Corollary \ref{cor:weakcluster}, to get 
\be
\Cov^{(0)}(f,g) \leq C ||f||_{2,(0),\rho} ||g||_{2,(0),\rho}e^{-\kappa d(\Sigma_f,\Sigma_g)}\ee

From Prop. \ref{equivgloClust}  the same statement holds in the symmetric covariance $\Cov$ and $\bL_2$ norm $||f||_{2,\rho}$. Then, invoking the equivalence between strong and weak clustering for one-dimensional systems  of Thm. \ref{thm:1Dequiv} we find that strong clustering holds, which in turn implies that all one-dimensional Gibbs samplers of commuting Hamiltonians are gapped via Thm. \ref{thm:main}. \qed
\end{proof}

\section{The High temperature phase}\label{sec:hightemp}

In this section we show that for $r$-local commuting Hamiltonians on a $d$-dimensional lattice there is a temperature $T_c(r, d)$, independent of the lattice size, such that for every $T \geq T_c$ both the Heat Bath and the Davies generators have a constant spectral gap. Thus the Gibbs state of every commuting Hamiltonian can be created efficiently on a quantum computer at high enough temperatures. The result follows from the mapping of Section \ref{sec:converse} between Liouvillians satisfying detailed balance and frustration-free Hamiltonians together with a technique due to Knabe \cite{Knabe} for lower bounding the spectral gap of local Hamiltonians. We show the result independently for Heat-Bath and for Davies generators. 

\begin{theorem}
Let $\Lambda$ be a subset of $\bZ^d$. Let $\Phi :\Lambda\mapsto \cA_\Lambda$ be a $r$-local bounded and commuting potential, and let $\rho$ be the Gibbs state of $H_\Lambda=\sum_{k \in \Lambda} \Phi(k)$. Then there exists a constant $T_c(r, d)$ such that for every $T \geq T_c$ the Heat Bath generator $\cL^H_\Lambda$ has a gap independent of $|\Lambda|$. 
\end{theorem}

\begin{proof}

By Eqs. (\ref{Cond_Exp}) and (\ref{eqn:HeatBath}), for any $f\in\cA_\Lambda$, 
\begin{equation}
- \hat{\cL}^H_\Lambda(f) = \sum_{k \in \Lambda} \rho^{1/4} (\id - \mathbb{E}^{\rho}_{k})(\rho^{-1/4} f \rho^{-1/4})\rho^{1/4}.
\end{equation}
where
\bea
\bE^{\rho_{T}}_{k}(f) &=& \Tr_{k}[\eta_k^\rho f \eta_k^{\rho \dag}] \\
&=&D_k \sum_{j} p_j \rho_{\not k}^{-1/4} U_{j} \rho_{\not k}^{-1/4}\rho^{1/2} g \rho^{1/2}  \rho_{\not k}^{-1/4}U_{j}^{\cal y} \rho_{\not k}^{-1/4},
\eea
with $\rho_{ \not k } := \text{tr}_{k} (\rho) $, $D_k$ the local Hilbert space dimension dimension, and $\{ p_j,  U_{j, k} \}$ an ensemble of depolarizing unitaries spanning $\cB_k$ such that for every $f\in\cB_k$, $\sum_{j} p_j U_j f U_{j}^{\cal y} = \text{tr}(f) \1_k / D_k$. 

Using the mapping of Section \ref{sec:converse}, the corresponding Hamiltonian $(-\hat{\cL}_\Lambda^H) $ acting on $\cB_\Lambda \otimes \overline{\cB_\Lambda}$ is given by 
\begin{eqnarray}
\hat{\cL}^{H}_\Lambda &:=& \sum_{k \in \Lambda} \hat{\cL}^{H}_k   \nonumber \\
&=& \sum_{k \in \Lambda} D_k(\rho^{1/4} \otimes \bar{\rho}^{1/4})  (\rho_{\not k}^{-1/4} \otimes \bar{\rho}_{\not k}^{-1/4})  \left(\1  -  \sum_{j} p_j U_{j, k} \otimes \bar{U}_{j, k} \right)  (\rho_{\not k}^{-1/4} \otimes \bar{\rho}_{\not k}^{-1/4}) ( \rho^{1/4} \otimes \bar{\rho}^{1/4}) \nonumber \\
&=& \sum_{k \in \Lambda} D_k(\rho^{1/4} \otimes \bar{\rho}^{1/4})  (\rho_{\not k}^{-1/4} \otimes \bar{\rho}_{\not k}^{-1/4})  \left(\1  -  w_{k\bar{k}} \right)  (\rho_{\not k}^{-1/4} \otimes \bar{\rho}_{\not k}^{-1/4}) ( \rho^{1/4} \otimes \bar{\rho}^{1/4}),
\end{eqnarray}
with $w_{k\bar{k}} = \sum_{i, j} \ket{j j} \bra{i i}$ the maximally entangled state  on $\cB_k\otimes\bar{\cB}_k$ tensored with the identity outside site $k$. 
Note that each $ \hat{\cL}^{H}_k$ is local, with its locality given by the interaction range $r$. Moreover as explained in Section \ref{sec:converse}, $\hat{\cL}^{H}_k$ is frustration free. 

Given that when $T \rightarrow \infty$, $\rho \rightarrow \1_\Lambda$ and so also $\rho_{\not k}^{-1/4} \rightarrow \1_\Lambda$, the Hamiltonian $(- \hat{\cL}^{H}_\Lambda)$ converges in the limit $T \rightarrow\infty$ to a non-interacting Hamiltonian given by $(- \hat{\cL}^{H}_\Lambda) = \sum_{k \in \Lambda}  (\id - w_{k \overline{k}})$, whose spectral gap is one. The statement of the theorem will follow by showing that there is a constant $T_c(k, d)$ such that the Hamiltonian $(- \hat{\cL}^{H}_\Lambda)$ has almost-commuting terms for all $T\geq T_c(k,d)$, with commutators sufficiently small that the spectral gap is also a constant. 

We proceed by employing a well-know technique due do Knabe \cite{Knabe} for lower bounding the spectral gap of local Hamiltonians. First we consider the Hamiltonian $(- \tilde{\cL}^{H}_\Lambda) := \sum_{k \in \Lambda} P_k^{H, T}$, where $P_k^{H, T}$ is the projector onto the non-zero eigenspace of $(- \hat{\cL}^{H}_k)$. 

We have that in the limit $T \rightarrow \infty$, $P_k^{H, T} \rightarrow (\id - w_{k \overline{k}})$. Moreover, we also have that 
\begin{equation} \label{connectingtosumofprojectors}
\Delta(- \hat{\cL}^{H}_\Lambda) \geq \Omega(\Delta( - \tilde{\cL}^{H}_\Lambda)).
\end{equation} 
(see e.g. Section 2 of \cite{DetectLemma2}). 

We now apply the Knabe bound \cite{Knabe}. It says that given any $k$-local frustration-free Hamiltonian on a $d$-dimensional lattice formed by local projector terms, then there is an integer $N(k, d)$ and a real number $\lambda(k, d) < 1$ (that can be computed explicitly given the lattice and that are independent of the volume) such that
\begin{equation}
\Delta(H) \geq  \Omega \left( \min_{S : |S| = N}   \Delta(H_{S}) - \lambda \right),\label{eqn:Knabe}
\end{equation}
where the minimum is taken over all connected sublattices of size $N$. 

For every fixed region $S$, $- \hat{\cL}^{H}_S$ converges to $\sum_{k \in S}  (\id - w_{k \overline{k}})$ in the limit $T \rightarrow \infty$. Since $\sum_{k \in S}  (\id - w_{k \overline{k}})$ has gap one, we find that given $N(k, d)$ and $\gamma(k, d)$, there always exists a $T_c$ such that for all $T \geq T_c$,
\begin{equation}
\min_{S : |S| = N}   \Delta(- \tilde{\cL}^{H}_S) > \lambda,
\end{equation}
so indeed
\begin{equation}
\Delta(- \tilde{\cL}^{H}_\Lambda) \geq \Omega(1),
\end{equation}
and the statement follows from Eq. (\ref{connectingtosumofprojectors}).\qed

\end{proof}

Similarly, we can show that the Davies generator is gapped at high temperature:

\begin{theorem}
Let $\Lambda$ be a finite subset of $\bZ^d$. Let $\Phi :\Lambda\mapsto \cA_\Lambda$ be an $r$-local bounded and commuting potential, and let $\rho$ be the Gibbs state of $H_\Lambda=\sum_{k \in \Lambda} \Phi(k)$. Then there exists a constant $T_c(r, d)$ such that for every $T \geq T_c$ the Davies generator $\cL^D_\Lambda$ has a gap independent of $|\Lambda|$. 
\end{theorem}

\begin{proof}

The proof is similar in spirit to the case of the of the Heat-Bath generator, so we will only expand on the aspects where the proofs differ. Direct evaluation shows that 
\bea \hat{\cL}^D_\Lambda&=& \sum_{k\in\Lambda,\alpha(k)}\sum_{\omega\geq0}\chi_{\alpha(k)}(\omega)e^{-\beta \omega/2}\left(S_{\alpha(k)}(\omega)\otimes \bar{S}_{\alpha(k)}(\omega)+S^\dag_{\alpha(k)}(\omega)\otimes \bar{S}^\dag_{\alpha(k)}(\omega)\right.\nonumber\\
&&-\half(S_{\alpha(k)}(\omega)S^\dag_{\alpha(k)}(\omega)\otimes \1+\1\otimes \bar{S}_{\alpha(k)}(\omega)\bar{S}^\dag_{\alpha(k)}(\omega))e^{\beta \omega/2} \nonumber \\
&&-\left.\half(S^\dag_{\alpha(k)}(\omega)S_{\alpha(k)}(\omega)\otimes \1+\1\otimes \bar{S}^\dag_{\alpha(k)}(\omega)\bar{S}_{\alpha(k)}(\omega))e^{-\beta \omega/2}\right)\eea

As $T\rightarrow \infty$, $-\hat{\cL}_\Lambda^D$ converges to 

\bea \hat{\cL}^D_\Lambda&=& -\sum_{k\in\Lambda,\alpha(k)}\sum_{\omega\geq0}\half\chi_{\alpha(k)}(\omega)\left(|S_{\alpha(k)}(\omega)\otimes\1-\1\otimes \bar{S}_{\alpha(k)}(\omega)|^2\right.\nonumber\\
&&+\left.|S^\dag_{\alpha(k)}(\omega)\otimes\1-\1\otimes \bar{S^\dag}_{\alpha(k)}(\omega)|^2\right)\eea

Now, recall that 

\be e^{-i t H}S_{\alpha(k)} e^{i t H}=\sum_\omega e^{i t \omega} S_{\alpha(k)}(\omega),\ee
hence the sum in $\omega$ is finite and independent of $|\Lambda|$, since $e^{-i t H}S_{\alpha(k)} e^{i t H}$ is local. 
We now want to use Parseval's theorem to show that $\hat{\cL}^D_\Lambda$ is gapped (for $\beta=0$). We will assume for simplicity that $\chi_{\alpha(k)}(\omega)$ is a constant independent of $\alpha,k,\omega$. 

\bea \hat{\cL}^D_\Lambda&=& -\half\chi\sum_{k\in\Lambda,\alpha(k)}\sum_{t\geq0}\left(|S_{\alpha(k)}(t)\otimes\1-\1\otimes \bar{S}_{\alpha(k)}(t)|^2\right.\nonumber\\
&&+\left.|S^\dag_{\alpha(k)}(t)\otimes\1-\1\otimes \bar{S^\dag}_{\alpha(k)}(t)|^2\right)\\
&=& 2\chi\sum_{k\in\Lambda,\alpha(k)}\sum_{t\geq0}e^{-i t H}\otimes e^{i t \bar{H}}(S_{\alpha(k)}\otimes \bar{S}_{\alpha(k)}-\1)e^{i t H}\otimes e^{-i t \bar{H}}\\
&=& 2\chi\sum_{t\geq0}e^{-i t H}\otimes e^{i t \bar{H}}\left(\sum_{k\in\Lambda,\alpha(k)}(S_{\alpha(k)}\otimes \bar{S}_{\alpha(k)}-\1)\right)e^{i t H}\otimes e^{-i t \bar{H}}\eea
where now the sum on $t$ is discrete and has a finite number of terms, by Parseval's theorem. Since the kernel of the central term is left unchanged by the time evolution, the gap of $\hat{\cL}_\Lambda^D$ is always larger than
\be \min_t {\rm gap}\left(e^{-i t H}\otimes e^{i t \bar{H}}\left(\sum_{k\in\Lambda,\alpha(k)}(S_{\alpha(k)}\otimes S_{\alpha(k)}-\1)\right)e^{i t H}\otimes e^{-i t \bar{H}}\right),\ee
but this is just a constant. Hence, the infinite temperature Davies map is gapped. 

Now invoking the same arguments leading up to Eqn. (\ref{eqn:Knabe}), we get that there exists a critical temperature, independent of the volume $|\Lambda|$, above which the Davies generators are gapped. 

\end{proof}

\section{Outlook}

We have introduced a unified framework for analyzing quantum Gibbs samplers of Hamiltonians with commuting local terms. This includes two independent prescriptions for constructing local quantum dynamical semigroups (i.e. Gibbs samplers) that uniquely drive the system to the Gibbs state of a given commuting Hamiltonian $H$. Associated to each Gibbs sampler, we construct local projectors onto the Gibbs state. The main result of the paper is a theorem which shows the equivalence between the rapid time convergence of the Gibbs sampler, and a new form of strong exponential clustering in the Gibbs state. We also explore how this new strong form of clustering is connected to more convensional notions of correlation decay. Finally, building upon the main theorem, we show that all Gibbs samplers of commuting Hamiltonians on a one dimensional lattice have a gap which is independent of the system size. Above a \mjk{system size independent} critical temperature, this holds true also for higher dimensional lattice models. 

These results are important and useful for a number of reasons. The two Gibbs samplers that we analyze serve complementary purposes in the literature. The Davies generators are meant to model the thermal dynamics that naturally emerge for a system weakly interacting with a thermal reservoir. This situation is very generic, especially for quantum optics based experiments, hence our analysis potentially provides crucial information on time scales for optical lattice simulators, and related setups. Secondly, the heat bath generators are a simple constructive semigroup which could be useful for quantum simulations. For certain tasks, they are  easier to work with than the Davies maps (ex: Ref.  \cite{ZiggyMaj1}). Finally, as outlined below, in the form of open questions, our main theorem provides a structural backbone relating several important notions, including: criticality, stability, topological order, classicality, etc.  

One major drawback of our framework is that it is not very well suited for Hamiltonians with non-commuting local terms. Indeed, it is easy to see that in general $\cL^D$, $\cL^H$, $\bE^\rho$, and $\bE^\cL$ all become non-local when $H$ is non-commuting, and very little of the framework can be recovered. It would be very interesting to explore extensions of our results to non-commuting Hamiltonians, as it would incorporate many of the more interesting models in quantum statistical mechanics. Sill in the setting of commuting Gibbs samplers, whether the spectral gap of the Liouvillian is equivalent to the Log-Sobolev inequality is an important open question. This equivalence holds for classical Gibbs samplers, and would allows a significant strengthening of Theorems \ref{thm:main} and \ref{thm:converse}. If one were able to extend the theory to Log-Sobolev inequalities, then it would be possible to show that Gibbs samplers have a relaxation time which is either exponential in the number of sites or logarithmic, i.e. there is no intermediate mixing regime \cite{Yoshida2}. 

Another very interesting direction to be explored in more detail is the connection between Gibbs samplers and frustration-free Hamiltonians outlined in Sec. \ref{sec:converse}. Many relevant problems in Quantum Hamiltonian complexity \cite{QHC1,QHC2} involve frustration-free Hamiltonians, and it is conceivable that by exploiting this new connection, the fields of quantum Gibbs samplers and Quantum Hamiltonian complexity can mutually benefit from their respective methods. In particular, it would be very interesting to understand to what extent the theory of Hypercontractive semigroups \cite{Hyper2,ZiggyHyper,LogSobolev} can be applied to problems of Hamiltonian complexity. \\

We conclude with a list of questions and conjectures, together with some compelling potential implications:

\bigskip
\textbf{1. The equivalence of weak and strong clustering in higher dimensions}

Theorem \ref{thm:1Dequiv} shows that for 1D systems, the strong and weak clustering conditions are equivalent, up to a multiplicative constant. Is this also true in higher dimensions? Although the proof of theorem \ref{thm:1Dequiv} clearly does not carry through to higher dimension because it relies heavily on a Schmidt decomposition of the boundary terms, there are reasons to believe that the equivalence could extend to two-dimensional lattice systems. 
Indeed, the conditional covariance, and the strong clustering condition, are an attempt to recover the situation when a state is clustering on a subset of the full lattice independently of the ``boundary conditions" that are chosen around the lattice restriction. In classical lattice systems, phase transitions can be driven along the boundary of a material whose bulk is in a thermally non-critical phase \cite{MartinelliReview}. Such a phenomenon has been coined a \textit{boundary phase transition}, and is believed to be a appear in quite natural models in three and higher spacial dimensions \cite{CzechModel}. However, in two-dimensional classical spin systems, this phenomenon cannot occur \cite{WeakMixingStrongMixing}. The heuristic reason for this is that the boundary of a two dimensional lattice model is effectively a one-dimensional lattice spin system, for which we know that no critical behaviour can be found. 

Hence it is tempting to conjecture the following: \textit{The Gibbs state of a commuting Hamiltonian on a two dimensional lattice satisfies weak clustering if, and only if, it satisfies strong clustering.}

\bigskip

\textbf{2. The behaviour of correlations as the temperature goes to zero}

Physicists study the ground states of Hamiltonians because in many situations it is believed that the actual state of the experiment is a Gibbs state at very low temperature, and the essential physics is governed by the properties of the ground state. The framework of Gibbs samplers provides a good setting for testing or confirming this intuition. In particular, if the ground state of a commuting Hamiltonian satisfies certain constraints on spacial correlations then one might expect that this still holds true at small non-zero temperature. 

We therefore raise the following questions: (i) if the ground state of a local commuting Hamiltonian satisfies clustering of correlations, does the same hold true for the Gibbs state at sufficiently low non-zero temperature? (ii) if a local commuting Hamiltonian satisfies LTQO, does the Gibbs sampler satisfy local indistinguishability at low temperatures? 

Partial answers can be given to these questions (in the operator norm), by using the approximation results in Ref. \cite{Has05a}, however a full answer is still elusive. One of the bottlenecks is that the in the Gibbs sampler setting, we have access to the machinery of $\bL_p$ norms whereas for (ground)-states the only natural norms are the operator norm for observables, and the inner product for states. Our theorems all depend on $\bL_2$ bounds, so a proper interpolation between zero temperature  and  finite temperature results is not obvious.  A candidate for an $\bL_p$ \mjk{norm-like object} on states is the following: for $f\in\cA_\Lambda$ and a pure state $\varphi\in\cS_\Lambda$, define $||f||^p_{p,\varphi}:= \bra{\varphi}|f|^p\ket{\varphi}$ with $1\leq p\leq \infty$. However, most interesting commuting Hamiltonian models, such as those exhibiting topological order, have a degenerate ground subspace, and with no preferential state it is hard to work with the \mjk{norm-like} $||f||_{p,\varphi}$. Hence, important obstacles still remain. 

\bigskip

\textbf{3. Absence of self-correction for 2D commuting Hamiltonians}

If the above two questions turn out to be true, then it would likely lead to a very strong result in the theory of self-correcting quantum memories: topological order of $2D$ commuting Hamiltonians on a lattice is unstable under thermal noise.

By self-correcting memory, we mean a Hamiltonian with a topologically stable ground subspace that remains a metastable subspace under thermal noise for a time that grows exponentially with the system size $L$. The thermal noise is usually modelled by Davies generators \cite{QMemory2,QMemory3}, or by a diagonal variant of it \cite{Loss}. The prototypical example of a genuinely self-correcting memory is the 4D toric code \cite{Preskill}. There have been a number of no-go theorems for self-correcting memories for 2 and 3 dimensional stabilizer codes \cite{Beni,BravyiTerhal,HastingsT0,PoulinNoGo}. All of the existing no-go theorems prove that under certain assumptions on the Hamiltonian, it only takes a constant amount of energy to flip from one logical eigenstate to another. According to the heuristic Arrhenius Law, which states that the survival time scales as an exponential of the free-energy barrier, this would prevent the system from being self-correcting. 

Arrhenius' law is known to be neither necessary nor sufficient in general for the existence of metastable states, so it would be desirable to have direct proofs that certain classes of Hamiltonians are not good quantum memories. Showing that the Davies generators of all 2D commuting Hamiltonians on a lattice are gapped would provide a definitive blow to self-correction in 2D, and would nicely complement the results in \cite{HastingsT0,PoulinNoGo,Kristan}.  

Assuming questions 1) and 2) are shown to be true, the argument for a no-go theorem would go as follows: 

\textit{If $H_\Lambda$  satisfies a specific form of Local Topological Quantum Order (LTQO), similar to the one defined in Ref. \cite{Angelo}, then the Davies generators $\cL^D$ are gapped, and hence, no state (quantum or classical) can survive for a time longer than polynomial in the system size. } 

We first define $LTQO_p$. Let $H_A$ be a Hamiltonian restricted to subset $A\in\Lambda$, and let $B\subseteq A$. If for any two ground states $\phi,\varphi$ of $H_A$, we have 
 \be |\avr{\phi| f|\phi}-\avr{\varphi| f|\varphi}|\leq c\avr{\phi| |f|^p|\phi}^{1/p} e^{-d(B,\partial A)/\xi},\label{eqn:LTQO}\ee
for any local observable $f\in\cA_B$,  and some constants $c,\xi>0$, then we say that $H_A$ satisfies $LTQO_p$. In particular all stabilizer hamiltonians satisfy $LTQO_p$ for all $p\geq 1$ since the RHS of Eqn. (\ref{eqn:LTQO}) is strictly zero beyond some constant distance. 

It is not difficult to see  that if a Hamiltonian satisfies $LTQO_p$, then its ground state is clustering in the following sense: there exist constants $c,\xi>0$ such that 
 \be |\bra{\phi} f g\ket{\phi}-\bra{\phi} f \ket{\phi}\bra{\phi} g\ket{\phi}| \leq c \avr{\phi| |f||\phi}\avr{\phi| |g|^p|\phi}^{1/p} e^{-d(\Sigma_f,\Sigma_g)/\xi}\label{eqn:ccp2}\ee
 
If one is then able to show that clustering in the form of Eqn. (\ref{eqn:ccp2}) also holds for non-zero temperature (i.e. question 2. above), then one would recover weak clustering in the Gibbs state for $p=2$. If in turn, weak and strong clustering are equivalent for 2D Gibbs samplers (question 1.), then one gets that topological order ($LTQO_2$) implies that the Gibbs sampler is gapped for all finite temperatures. 

This type of reasoning, sketchy at this point, shows the power of our main theorems (Thms. \ref{thm:main} and \ref{thm:converse}) in terms of relating static and dynamical properties of spin systems in thermal equilibrium.\\

\textbf{Acknowledgements}: We gratefully acknowledge fruitful discussions with T.~Cubitt, J.~Eisert, M.~Friesdorf, A.~Lucia, F.~Pastawski, K.~Temme, and R.~F.~Werner. MJK was supported by the Alexander von Humboldt foundation and by the EU (SIQS, RAQUEL). FGSLB was supported by an EPSRC Early Career fellowship. 

%%%%%%%%%%%%%%%%%%%%%%%%%%%%%%%%%%%%%%%%%%%%%%

%%%%%%%%%%%%%%%%%%%%%%%%%%%%%%%%%%%%%%%%%%%%%%%%%%%%%%

\bibliographystyle{apsrev}

\section*{Appendix}

\textbf{Proof of Lemma \ref{lemma:gap}:}
We consider the expression for $\lambda_\Lambda(A)$, and note that

\bea \lambda_\Lambda(A)&=& \inf_{f\in\cA_\Lambda}\frac{-\tr{\Gamma_{\rho}(f)\cL^D_A(f)}}{\tr{\cQ_A(f)\Gamma_{\rho}(\cQ_A(f))}}\\
&=& \inf_{g\in\cA_\Lambda}\frac{-\tr{g\tilde{\cL}^D_A(g)}}{\tr{\tilde{\cQ}_A(g)\tilde{\cQ}_A(g)}},\label{geneigen}\eea
where we made the replacement $g\equiv\Gamma^{1/2}_{\rho}(f)$, and  $\Gamma_\rho (f)=\rho^{1/2}f\rho^{1/2}$. We defined the operators $\cQ_A(f):=f-\bE_A(f)$, $\tilde{\cL}^D_A=\Gamma^{1/2}_{\rho}\cL^D_A \Gamma^{-1/2}_{\rho}$, and likewise  $\tilde{\cQ}_A=\Gamma^{1/2}_{\rho}\cQ_A \Gamma^{-1/2}_{\rho}$. We note that $\tilde{\cL}^D_A$ and $\tilde{\cQ}_A$ are both hermitian operators, so Eqn. (\ref{geneigen}) is a generalized eigenvalue equation which can be recasted as
\be \lambda_\Lambda(A) = \inf\{\lambda | \det(\tilde{\cL}^D_A+\lambda \tilde{\cQ}_A^2)=0\}\label{detformula}\ee

Now, we will show that $\tilde{\cQ}_A$ acts non-trivially only on $A_\partial$ (the same is true for the pair  $(\bE^\rho,\cL^H_\Lambda)$). We will write the subscript of $\rho$ explicitly so as to avoid confusion. 
Indeed, for any $g\in\cA_\Lambda$,
\bea \tilde{\cQ}_A(g)&=& g - \Gamma_{\rho_\Lambda}^{1/2}(\bE^\cL_A(\Gamma_{\rho_\Lambda}^{-1/2}(g)))\\
&=& g - \rho_\Lambda^{1/4}\bE_A^\cL(\rho_\Lambda^{-1/4}g\rho_\Lambda^{-1/4})\rho_\Lambda^{1/4}\\
&=& g- \rho_{A_\partial}^{1/4}\rho_{(A_\partial)^c}^{1/4}\bE_A^\cL(\rho_{(A_\partial)^c}^{-1/4} \rho_{A_\partial}^{-1/4}g\rho_{A_\partial}^{-1/4}\rho_{(A_\partial)^c}^{-1/4})\rho_{(A_\partial)^c}^{1/4} \rho_{A_\partial}^{1/4}\\
&=& g - \Gamma_{\rho_{A_\partial}}^{1/2}(\bE_A^\cL(\Gamma_{\rho_{A_\partial}}^{-1/2}(g)))\\
&=& \tilde{\Psi}_{A_\partial}\otimes {\rm id}_{(A_\partial)^c}(g)\eea
for some hermitian operator $ \Psi_{A_\partial}$ acting only on $A_\partial$. Similarly, we get that $\tilde{\cL}^D_A$ can be written as $\tilde{\cK}_{A_\partial}\otimes\1_{\bar{A_\partial}}$ for some hermitian operator $\tilde{\cK}_{A_\partial}$. 
Then, Eqn. (\ref{detformula}) can be rewritten as 
\be \lambda_\Lambda(A) = \inf\{\lambda | \det((\tilde{\cK}_{A_\partial}+\lambda \tilde{\Psi}_{A_\partial}^2)\otimes {\rm id}_{(A_\partial)^c})=0\}\ee
Recalling now that $\det(A\otimes \1)=\det(A)^n$, where $n$ is the dimension of matrix $A$, we get that 
\bea \lambda_\Lambda(A) &=& \inf\{\lambda | \det((\tilde{\cK}_{A_\partial}+\lambda \tilde{\Psi}_{A_\partial}^2)\otimes {\rm id}_{(A_\partial)^c})=0\}\\
&=& \inf\{\lambda | \det(\tilde{\cK}_{A_\partial}+\lambda \tilde{\Psi}_{A_\partial}^2)=0\}=\lambda_{A_\partial}(A),\eea
completing the proof.
\qed

\bigskip

\textbf{Proof of Lemma \ref{lem:iteration}:}
Let $C:=[a_1,b_1]\times...\times[a_d,b_d]\in\cR_k^d/\cR_{k-1}^d$. We can assume that $a_n=0$ and $b_n\leq l_{k+n}$ for $n=1,...,d$. Then necessarily $b_d>l_k$, since otherwise $C\in \cR_{k-1}^d$. Define
\bea
A_i&:=& [0,b_1]\times...\times[0,b_{d-1}\times[0,\frac{b_d}{2}+\frac{2i}{8}\sqrt{l_k}],\\
B_i&:=&[0,b_1]\times...\times[0,b_{d-1}]\times[\frac{b_d}{2}+\frac{21-1}{8}\sqrt{l_k},b_d]\eea
We have $d(C/A-i,C/B-1)=\frac{1}{8}\sqrt{l_k}$. Furthermore,
\be \frac{b_d}{2}+\frac{2 s_k}{8}\sqrt{l_k}\leq \frac{l_{d+k}}{2}+\frac{1}{4}l_k^{5/6}\leq\frac{3l_k}{4}+\frac{1}{4}l_k^{5/6}\leq l_{k-1+d}\ee
which together with the fact that $l_k<b_d$, implies that $A_i$ and $B_i$ are both subsets of $C$. Moreover, since for all $i=1,...,s_k$ 
\be \frac{b_d}{2}+\frac{2i}{8}\sqrt{l_k}\leq l_k,~~~~b\leq l_{k+1},...,b_{d-1}\leq l_{k-1+d}\ee
we find that $A_i$ belongs to $\cR_{k-1}^d$. The sets $B_i$'s also belong to $\cR_{k-1}^d$, since they are smaller than the $A_i$'s. \qed

\end{document}